\documentclass[prx,twocolumn,english,superscriptaddress,floatfix,longbibliography]{revtex4-2}

\usepackage[utf8]{inputenc}
\usepackage{amsmath, graphicx, amsfonts, amsthm, algorithm, algpseudocode}
\usepackage{array}
\usepackage{tikz}
\usetikzlibrary{quotes,positioning,arrows.meta}
\usepackage[breaklinks=true]{hyperref}
\usepackage[normalem]{ulem}
\usepackage{xcolor}
\usepackage{mathtools}
\usepackage{fixdif}
\usepackage{subcaption}

\usepackage[compat=newest]{yquant}
\useyquantlanguage{groups}
\yquantset{
   operators/every box/.append style={}, % rounded corners
   operators/every control box/.style={/yquant/operators/every rectangular box},
}
\makeatletter
\yquant@langhelper@declare@command@uncontrolled%
   {controlbox}%
   \yquant@register@get@allowmultitrue%
   {%
      \expandafter\yquant@prepare@injection%
         \expandafter{\yquant@lang@attr@value}%
         {/yquant/operators/every control box}%
   }
\makeatother

\newtheorem{definition}{Definition}
\newtheorem{theorem}{Theorem}
\newtheorem{theoreminformal}{Theorem}
\newtheorem{lemma}{Lemma}

\newcommand{\braket}[1]{\left\langle #1 \right\rangle}
\newcommand{\bra}[1]{\left\langle #1 \right|}
\newcommand{\ket}[1]{\left| #1 \right\rangle}

\newcommand{\tracenorm}[1]{\left\|#1\right\|_1}
\newcommand{\opnorm}[1]{\left\|#1\right\|}

\newcommand{\tr}{\operatorname{tr}}

\newcommand{\mA}{\mathcal{A}}
\newcommand{\mB}{\mathcal{B}}
\newcommand{\mC}{\mathcal{C}}

\newcommand{\mH}{\mathcal{H}}
\newcommand{\mI}{\mathcal{I}}

\newcommand{\mL}{\mathcal{L}}

\newcommand{\mR}{\mathcal{R}}
\newcommand{\mS}{\mathcal{S}}

\newcommand{\E}{\mathbb{E}}
\renewcommand{\P}{\mathbb{P}}
\begin{document}

\title{Stochastic Quantum Hamiltonian Descent}

\author{Sirui Peng}
\affiliation{State Key Lab of Processors, Institute of Computing Technology, Chinese Academy of Sciences, 100190, Beijing, China.}
\affiliation{School of Computer Science and Technology, University of Chinese Academy of Sciences, Beijing, 100049,  China.}

\author{Shengminjie Chen}
 \affiliation{State Key Lab of Processors, Institute of Computing Technology, Chinese Academy of Sciences, 100190, Beijing, China.}

 \author{Xiaoming Sun}
 \affiliation{State Key Lab of Processors, Institute of Computing Technology, Chinese Academy of Sciences, 100190, Beijing, China.}

 \author{Hongyi Zhou}
\email{zhouhongyi@ict.ac.cn}
 \affiliation{State Key Lab of Processors, Institute of Computing Technology, Chinese Academy of Sciences, 100190, Beijing, China.}

\date{\today}

\begin{abstract}
Stochastic Gradient Descent (SGD) and its variants underpin modern machine learning by enabling efficient optimization of large-scale models. However, their local search nature limits exploration in complex landscapes. In this paper, we introduce Stochastic Quantum Hamiltonian Descent (SQHD), a quantum optimization algorithm that integrates the computational efficiency of stochastic gradient methods with the global exploration power of quantum dynamics. We propose a Lindbladian dynamics as the quantum analogue of continuous-time SGD. We further propose a discrete-time gate-based algorithm that approximates these dynamics while avoiding direct Lindbladian simulation, enabling practical implementation on near-term quantum devices. We rigorously prove the convergence of SQHD for convex and smooth objectives. Numerical experiments demonstrate that SQHD also exhibits advantages in non-convex optimization. All these results highlight its potential for quantum-enhanced machine learning.
\end{abstract}

\maketitle

\section{Introduction}\label{sec:intro}
Stochastic Gradient Descent (SGD)~\cite{Robbins1951ASA} and its variants are the predominant optimization algorithms for large-scale machine learning.
In particular, SGD aims to minimize the objective function $f(x)=\frac{1}{m}\sum_{j=1}^m f_j(x)$ with access to individual functions $f_1(x),\ldots,f_m(x)$, where $m$ corresponds to the size of the training dataset. Unlike Gradient Descent (GD), which computes updates using the entire dataset, SGD leverages small random subsets of data to achieve superior computational efficiency. This stochastic approach provides significant advantages in modern data-intensive applications where full-batch computations are prohibitively expensive.

The gradient-based optimization and the dynamical system have deep connections \cite{Wibisono2016AVP,MZKY23,SSJ2023}. In the quantum domain, Quantum Hamiltonian Descent (QHD)~\cite{QHD} has emerged as a distinct optimization paradigm. By exploiting quantum tunneling effects, QHD enables non-local exploration of objective function landscapes, allowing it to traverse energy barriers and discover high-quality solutions inaccessible to classical gradient-based methods~\cite{QHD-separation-1,QHD-separation-2}. However, QHD's continuous-time evolution requires full dataset queries, making it computationally prohibitive for large-scale problems. Indeed, tackling the challenges of scalability and trainability in quantum optimization is an active area of research \cite{Abbas2024}, and there have been significant advances in the solution of critical problems, including Semidefinte Programming (SDP) \cite{QSDP-1,QSDP-2,QSDP-3,QSDP-4} , Linear Programming (LP) \cite{QLP-1,QLP-2,QLP-3,QLP-4,QLP-5} and more generalized continuous optimization problem \cite{QAL,QCONVEX-1,QNONCONVEX-1}.

The computational cost associated with QHD's requirement for full dataset access remains a significant barrier to its application in large-scale, data-intensive machine learning tasks, precisely where classical SGD excels by using efficient, stochastic mini-batch updates. This raises a fundamental question: \textit{Can we develop a quantum counterpart to SGD that combines the computational efficiency of stochastic sampling with the global search capabilities of quantum optimization?}
A key observation is that the iterative process of SGD can be modeled as a dynamical system subjected to stochastic forces originating from environmental interactions. This inspires us to construct an open quantum system that parallels the structure of SGD. Specifically, the random choice of individual functions in the gradient calculation of SGD corresponds to random potentials in an open quantum system. As a result, the open system is driven by stochastic forces during its evolution. The evolution of such a quantum system is described by the Lindblad master equation, where the coupling to the environment manifests itself as dissipation terms. The dissipation can be engineered to emulate the stochastic noise in SGD, guiding the system towards global minima. Simultaneously, the inherent quantum effects, including tunneling, enable non-local transitions that enhance global exploration.

In this work, we introduce \textit{Stochastic Quantum Hybrid Dynamics}, a continuous-time open quantum system that provably converges on convex and smooth objective functions. Based on this, we propose \textit{Stochastic Quantum Hamiltonian Descent} (SQHD), an efficient discrete-time quantum algorithm that approximates the continuous-time open system dynamics. This algorithm combines the global exploration capabilities of QHD with the computational efficiency of SGD. Our central contributions are formalized into two theorems: \textbf{Theorem~\ref{thm:order-2-conv-informal}} establishes the convergence of the continuous-time dynamics, and \textbf{Theorem~\ref{thm:order-2-approx-informal}} demonstrates that the discrete-time algorithm accurately approximates the continuous-time dynamics. Together, these theorems ensure the convergence of the SQHD algorithm. By exploiting the structured stochasticity of the continuous-time dynamics, our gate-based quantum algorithm circumvents the need for direct Lindblad simulations, making it suitable for near-term quantum devices. We complement our theoretical findings with numerical experiment demonstrations, showing that SQHD also exhibits advantages in non-convex optimization, positioning it as a promising approach within the quantum machine learning landscape.

\section{Result}\label{sec:result}
We first briefly overview our main contributions.
\begin{itemize}
    \item \textbf{Stochastic Quantum Hybrid Dynamics and its convergence analysis}: First, we propose Stochastic Quantum Hybrid Dynamics, a real-space quantum dynamical system depicted by the Lindblad master equation, and the system is guaranteed to evolve towards the global minimum on any convex objective function under mild assumptions.
    \item \textbf{Stochastic Quantum Hamiltonian Descent as an approximation}: Next, we propose a gate-based quantum algorithm, Stochastic Quantum Hamiltonian Descent, which well approximates the trajectory of the proposed dynamical system. The algorithm utilizes structured stochasticity in the dynamical system and avoids simulating the Lindblad dynamics directly.
    \item \textbf{Numerical results on diverse nonconvex landscapes}: We compare our algorithm with Stochastic Gradient Descent with Momentum (SGDM) and Quantum Hamiltonian Descent (QHD) on five non-convex objective functions with different features. The results demonstrate that our algorithm maintains the non-local exploration ability that helps to escape from local minima, even for objective functions with a large gradient noise.
\end{itemize}

More technical details are covered in Appendix, including detailed problem settings and assumptions, formal statements and proofs of the convergence result and the approximation result. The complete information on the numerical experiment can also be found in Appendix~\ref{sec:full-numerical}.

\begin{figure*}
    \centering
    \includegraphics[width=\textwidth]{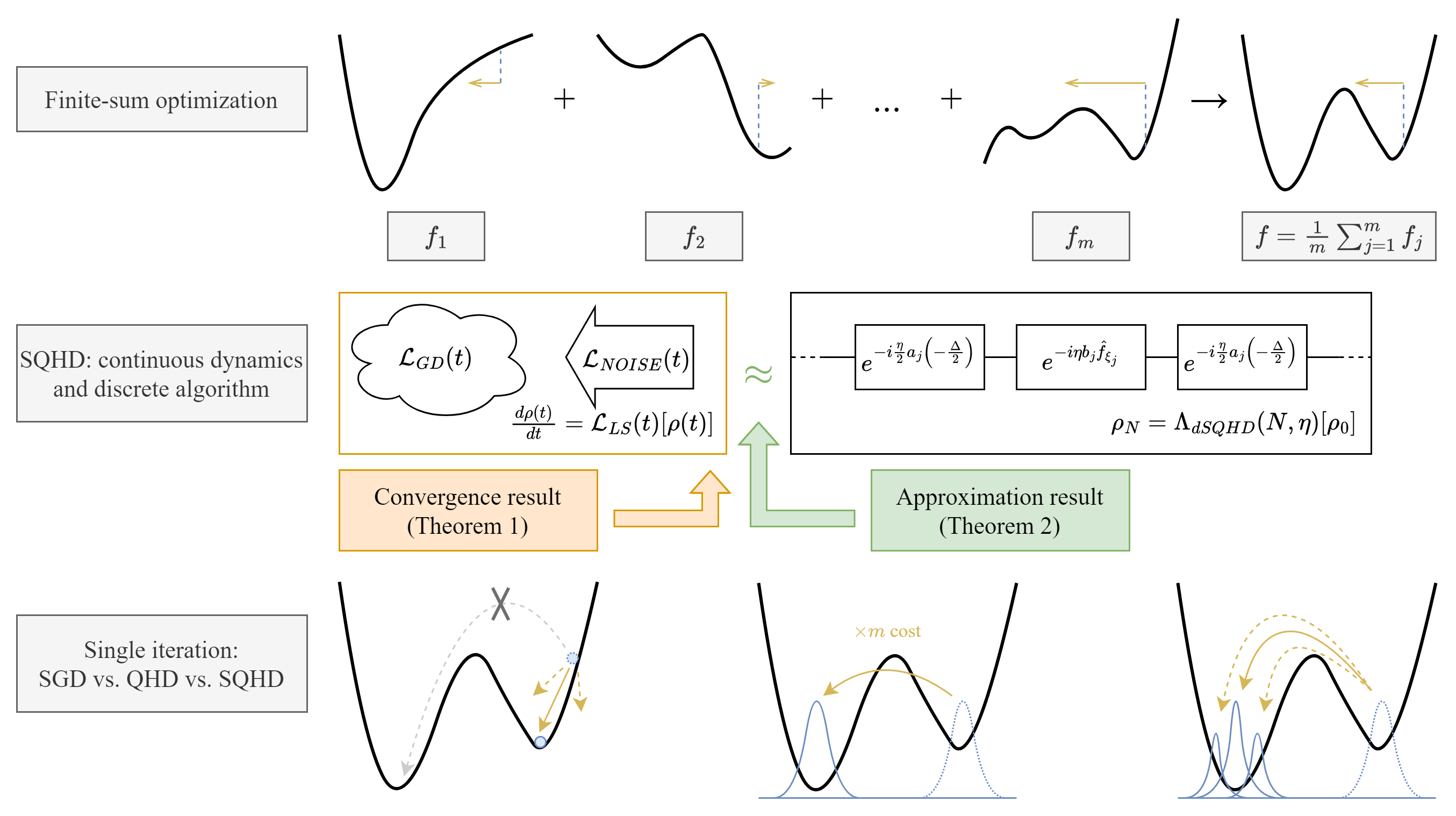}
    \caption{Overview of the SQHD method. (Top) The SQHD method is designed for unconstrained finite-sum optimization problems, where queries to individual objective functions $f_1,\ldots,f_m$ (rather than the full objective function $f$) are available. (Middle) The abbreviation carries a dual meaning: Stochastic Quantum Hybrid Dynamics, a quantum dynamical system that is guaranteed to converge on convex and smooth objective functions, and Stochastic Quantum Hamiltonian Descent, a tailored gate-based quantum algorithm that is guaranteed to simulate the dynamical system efficiently. To avoid ambiguity, in the rest of the paper we use SQHD exclusively for the quantum algorithm. (Bottom) It can be prohibited to escape from local minima in SGD while the QHD method requires multiple queries to the objective function. Our method combines the best of both worlds: the solution quality of QHD and the computational efficiency of SGD, and we demonstrate this argument both theoretically and numerically.}
    \label{fig:SQHD-schema}
\end{figure*}

\subsection{Convergence of Stochastic Quantum Hybrid Dynamics}

The gradient descent process can be conceptualized as the trajectory of a particle in a potential energy landscape specified by the objective function. This inherent connection between gradient descent methods and dynamical systems provides a rich theoretical framework \cite{Wibisono2016AVP}, offering insights into the behavior of gradient-based optimization techniques. Such a correspondence also extends to closed quantum dynamical systems \cite{QHD}. One of our key contributions is expanding upon this established correspondence by drawing a novel parallel between stochastic gradient methods and open quantum systems.

The key challenge is to characterize the effect of using stochastic gradients. In our quantum setting, the query to the total objective function $f$ corresponds to evolving the quantum state $\rho$ with the phase operator $e^{-i\eta \hat{f}}$ where $\hat f=\int_x f(x)\ket{x}\bra{x}\d x$. The stochastic gradient update, however, involves applying a randomly selected operator, $e^{-i\eta \hat{f}_\xi}$, where $\hat{f}_\xi$ is chosen uniformly from a set of components $\{\hat{f}_1, \ldots, \hat{f}_m\}$. The key insight comes from comparing the Taylor series expansions of these two types of evolution. To the first order, both methods produce the same evolution, $-i\eta [\hat{f},\rho]$, meaning the stochastic approach, on average, steers the system in the same direction as the deterministic approach. The difference---and the characteristic effect of stochasticity---appears in the second-order term. By isolating the difference between the second-order terms of the full and the averaged stochastic evolutions, we can precisely quantify the noise introduced by the stochastic process.
\begin{align*}
\text{Stochastic Noise} \approx \frac{\eta^2}{2}\left([\hat{f},[\hat{f},\rho]] - \frac{1}{m}\sum_{j=1}^m [\hat{f}_j,[\hat{f}_j,\rho]]\right).
\end{align*}
This expression captures the distinctive diffusive effect that arises from using random components rather than the full operator.

Based on this analysis, we propose the following dynamical system to model the continuous-time evolution of the stochastic approach:

\begin{definition}[Stochastic Quantum Hybrid Dynamics]
The evolution of the mixed quantum state $\rho(t)$ is described by:
\begin{align}\label{eqn:LME-for-SQHD}
\frac{\d \rho(t)}{\d t} = \mL_{LS}(t)[\rho(t)],    
\end{align}
where the generator $\mL_{LS}$ is composed of two distinct parts:
\begin{enumerate}
\item \textbf{A Gradient Descent Term ($\mL_{GD}$):} This term drives the system along the desired optimization path according to the standard Schroedinger equation, governed by the Hamiltonian $H(t)=e^{\psi(t)}\left(-\frac{1}{2}\Delta\right)+e^{\chi(t)}\hat{f}$, where $\Delta=\frac{\partial^2}{\partial x_1^2}+\frac{\partial^2}{\partial x_2^2}+\cdots+\frac{\partial^2}{\partial x_d^2}$ is the Laplace operator. This represents the ideal, noiseless component of the dynamics,
\begin{align*}
  \mL_{GD}(t)[\sigma] = -i[H(t),\sigma].    
\end{align*}
\item\textbf{A Stochastic Noise Term ($\mL_{NOISE}$):} This term explicitly models the diffusive effects introduced by the stochastic updates, as derived in our analysis above.
\begin{align*}
  \mL_{NOISE}(t)[\sigma] = \frac{e^{2\chi(t)}}{2}\left([\hat{f},[\hat{f},\rho]]-\frac{1}{m}\sum_{j=1}^m[\hat{f}_j,[\hat{f}_j,\rho]]\right).
\end{align*}
\end{enumerate}
The overall dynamics are then a weighted combination of these two effects, modulated by learning rate $\eta$ and learning rate schedule $u(t):[0,T]\to[0,1]$:
\begin{align*}
\mL_{LS}(t) = u(t)\mL_{GD}(t) + u(t)^2\eta\mL_{NOISE}(t).
\end{align*}
The corresponding evolution channel for time period $[0,T]$ is $\Lambda_{LS}(u,0,T)$. When $u(t)=1,t\in[0,T]$, the channel is simply denoted as $\Lambda_{LS}(0,T)$.
\end{definition}
This formulation provides a clear and accurate model of the system's trajectory, capturing both the intended descent and the characteristic impact of the underlying stochastic process. The process described above is a valid quantum channel, since Eq.~\eqref{eqn:LME-for-SQHD} is a Lindblad master equation.

Stochastic Quantum Hybrid Dynamics is guaranteed to converge on smooth and convex objective functions, for any smooth initial state (see Definition 1 in Appendix~\ref{def:smooth}). The smoothness of a real-space quantum state means its corresponding wavefunction is also smooth, varying gently across space. This condition aims to exclude unrealistic states with infinite kinetic energy, as kinetic energy is related to the wavefunction's second derivative.
\begin{theoreminformal}[informal]
\label{thm:order-2-conv-informal}
Assume that $f=\frac{1}{m}\sum_{j=1}^m f_j$ is the sum of functions $\{f_1,\ldots,f_m\}$, and $f_j$ are convex and smooth functions.
Let $\sigma_f^*=\E_{j}[\opnorm{\nabla f_j (x^*)-\nabla f(x^*)}^2]$ be the gradient noise of the function $f$, where $\opnorm{\cdot}$ denotes $\ell^2$ norm for a vector.
Let $x^*$ be the unique local minimizer of $f$.

For any smooth initial state $\rho_0$ and functions $\alpha,\beta,\gamma:[0,T]\to\mathbb{R}$ that satisfies the strong ideal scaling condition
\begin{align*}
\psi(t) &= \alpha(t)-\gamma(t), \\
\chi(t) &= \alpha(t)+\beta(t)+\gamma(t), \\
\dot{\beta}(t) &= \dot{\gamma}(t) = e^{\alpha(t)},
\end{align*}
the result state $\rho_t$ of Stochastic Quantum Hybrid Dynamics \eqref{eqn:LME-for-SQHD} with learning rate schedule $u(t)=e^{-(\alpha(t)+\beta(t))}$ satisfies
\begin{align}\label{eqn:bound-order-2-conv}
\braket{\hat{f}}_t - f(x^*) = O\left(e^{-\beta(\tau)}+\eta\sigma_f^*\right),
\end{align}
where $\braket{\hat f}_t=\tr(\hat f\rho_t)$ and $\tau = \int_0^t u(s)\d s$ is the effective time.
\end{theoreminformal}
The bound for the expected loss $\braket{\hat{f}}_t-f(x^*)$ in Theorem \ref{thm:order-2-conv-informal} consists of two parts: the first part shows a fast descent behavior at the initial stage, and the second part is linear to the learning rate $\eta$ and gradient noise $\sigma_f^*$, showing a fluctuation behavior in long-time limit.

Theorem \ref{thm:order-2-conv-informal} provides a systematic depiction of the dynamical system's behavior on convex and smooth objective functions, which captures features that may be helpful to practical applications. For instance, the separation of the quantum descent phase and the classical fluctuation phase, $t^*$, is obtained when the two components in Eq.\eqref{eqn:bound-order-2-conv} are of the same order. If we choose 
    \begin{align*}
    \alpha(t)&= -\log (t+t_\epsilon), \\
    \beta(t) &=\log (t+t_\epsilon) +\log C,C>1 \\
    \gamma(t) &=\log (t+t_\epsilon),\\
    t_\epsilon&>0,t_\epsilon\to 0,
    \end{align*}
in correspondence to stochastic momentum method \cite[Theorem 7.4]{handbook}, then $\braket{\hat{f}}_t-f(x^*)= O(\frac{1}{t}+\eta\sigma_f^*)$, the separation is $t^*\sim\eta^{-1}$.

\subsection{Stochastic Quantum Hamiltonian Descent as an Approximation}
We assume access to the standard evaluation oracle
\begin{align}
O_{f_j}\ket{x,z}=\ket{x,f_j(x)+z}, x\in\mathbb{R}^d,z\in\mathbb{R},
\end{align}
which queries function $f_j$ coherently for $j=1,\ldots,m$.
The most direct approach to running our method on a digital quantum computer is by resorting to the general Lindbladian simulation algorithm \cite{CW16,LW23,DLL24,PSZZ24}. These algorithms query the Hamiltonian and jump operators of the system, and, critically, require querying the full objective function $f$. For large datasets, this can be computationally prohibitive. Besides, the dissipative terms in the dynamical system correspond to complex environmental interactions, requiring a large-scale quantum circuit infeasible for current devices. 
To address this limitation, we propose Stochastic Quantum Hamiltonian Descent, a tailored quantum algorithm designed to circumvent the need for querying the entire dataset.

\begin{algorithm}[H]
\caption{Stochastic Quantum Hamiltonian Descent}\label{alg:SQHD}
\begin{algorithmic}
\Require Learning rate $\eta$, iteration number $N$, coefficient functions $\psi(t),\chi(t)$, and evaluation oracles $O_{f_j},j=1,\ldots,m$.
\begin{enumerate}
    \item Prepare the initial guess state $\rho_0$.
    \item For epoch $j=0,1,\ldots,N-1$:
    \begin{enumerate}
        \item Calculate the discrete parameters $a_j=\exp(\psi((j+1/2)\eta))$, $b_j=\exp(\chi((j+1/2)\eta))$.
        \item Update with unitary $\exp(-i \frac{\eta}{2} a_j (-\Delta/2))$.
        \item Update with unitary $\exp(-i\eta b_j \hat{f}_{\xi_j})$, where $\xi_j$ is independently and uniformly drawn from $\{1,\ldots,m\}$.
        \item Update with unitary $\exp(-i \frac{\eta}{2} a_j (-\Delta/2))$.
    \end{enumerate}
    \item Measure the final state with the position operator $\hat x=\int_x x\ket{x}\bra{x}\d x$ and output the measured value.
\end{enumerate}    
\end{algorithmic}
\end{algorithm}

This algorithm outputs a value $x$ such that $\E[f(x)]$ is approximately $\braket{\hat f}_t$ where $t=N\eta$. The update $U_{dSQHD}(N,\eta;\xi)$ (Step \textbf{2} of the algorithm) is similar to a second-order Trotter-Suzuki decomposition of the Hamiltonian dynamic in the Lindblad equation, and the difference is that the query to $\hat f$ is replaced by a stochastic query to individual $\hat f_j,j=1,\ldots,m$. The update can be implemented efficiently using techniques from real-space dynamics simulation \cite{Childs2022QuantumSO}.

The stochastic update in the algorithm is
\small
\begin{align}
\label{eqn:discrete-SQHD}
\Lambda_{dSQHD}(N,\eta)[\rho] = \E_{\xi}\left[U_{dSQHD}(N,\eta;\xi)\rho U^\dagger_{dSQHD}(N,\eta;\xi)\right],
\end{align}
\normalsize
and $\Lambda_{dSHQD}(k,\eta)$ is a good approximation for $\Lambda_{LS}(0,k\eta)$ for $k=1,\ldots,N$ for smooth quantum states.

\begin{theoreminformal}[informal]
\label{thm:order-2-approx-informal}
Given iteration number $N$ and learning rate $\eta\in(0,1)$. For any initial state $\rho_0$ such that $\tilde \rho_t=\Lambda_{LS}(0,t)[\rho_0]$ is smooth for $t\in[0,N\eta]$, and any smooth functions $e^{\psi(t)},e^{\chi(t)}$ and $f_j(x),j=1,\ldots,m$, the Stochastic Quantum Hybrid Dynamics process $\tilde \rho_t$ is an order-$2$ quantum weak approximation of the Stochastic Quantum Hamiltonian Descent process $\rho_k=\Lambda_{dSQHD} (k,\eta)[\rho_0]$.
\end{theoreminformal}
Theorem \ref{thm:order-2-approx-informal} rigorously connects the discrete gradient descent process \eqref{eqn:discrete-SQHD} to the continuous dynamical system \eqref{eqn:LME-for-SQHD}. The approximation result also holds for an adaptive learning rate (see Appendix~\ref{sec:proof-approx}).
The approximation result implies that the continuous-time dynamics can be realized using methods analogous to a Trotter-Suzuki decomposition, making it a promising candidate for execution on near-term quantum devices. In contrast, a direct simulation of the Lindblad dynamics, which is required for other approaches such as Quantum Langevin Dynamics for optimization \cite{QLD}, is often computationally prohibitive for large datasets and would require simulating complex environmental interactions, demanding large-scale quantum circuits that are not feasible on current hardware. The proposed SQHD algorithm circumvents these challenges by avoiding direct Lindbladian simulation.

\subsection{Numerical results}
\label{sec:numerical}
We consider unconstrained finite-sum optimization problems with five non-convex objective functions: Styblinski-Tang function (\texttt{dw}), Michalewicz function (\texttt{mich}), Cube-Wave functions (\texttt{cubewave}), and two functions from the Nonlinear Least Squares problem (\texttt{sino,sino-alt}). These functions are common for testing optimization problems and cover diverse landscapes.
We evaluate the performance of SQHD against QHD and SGDM on these problems by comparing their expected loss and the probability of successfully finding the global minimum. The results demonstrate that SQHD achieves comparable or even superior solution quality to QHD on these non-convex problems.

\begin{figure*}
    \centering
    \includegraphics[width=\textwidth]{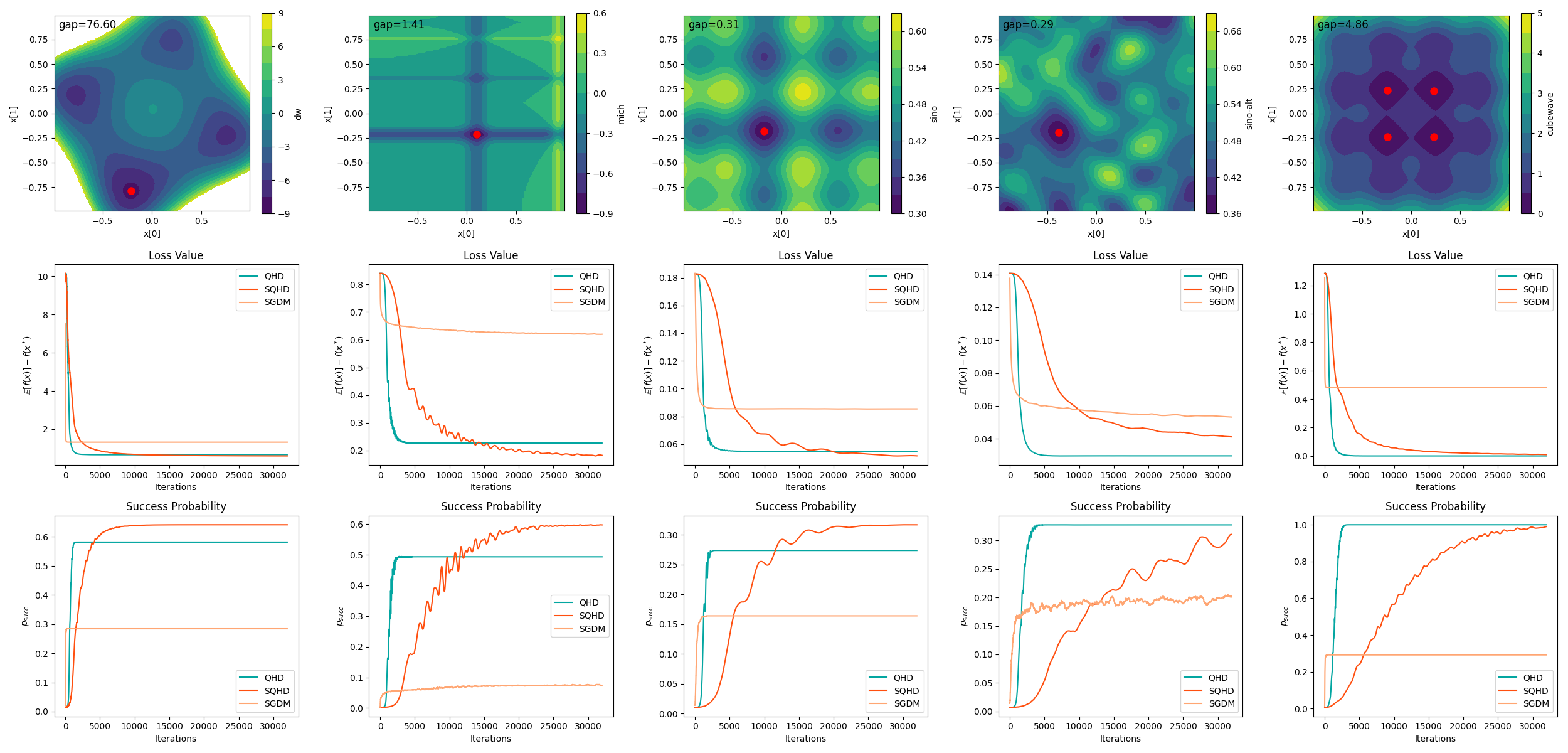}
    \caption{Comparison of SQHD, QHD, and SGDM on objective functions \texttt{dw,mich,sino,sino-alt,cubewave} under default parameter settings (see Appendix~\ref{sec:full-numerical}), demonstrating SQHD's ability to escape from local minimum. In the first row, the contour maps of these objective functions are shown, showing different landscape features. We use the expected loss $\E[f(x)]-\min f$ and $\delta$-success probability $p_{succ}=P\left(\frac{f(x)-\min f}{\max f-\min f}<\delta\right)$ to evaluate the performance of different algorithms, which are shown in the second row and the third row.}
    \label{fig:general-comparison}
\end{figure*}

We also investigate the effect of different parameter settings, including learning rate $\eta$ and resolution of space discretization. We find that a smaller learning rate usually improves the solution quality, especially for problems with large gradient noise. While the discretization resolution does not affect the solution quality of SQHD, a larger resolution can lead to a slower convergence. This is significantly different from QHD, whose solution quality and convergence speed are consistent in different resolutions. Details of these results can be found in the Appendix~\ref{sec:full-numerical}.

\section{Method}\label{sec:method}
\subsection{Quantum Hamiltonian Descent}
In \cite{QHD}, the authors proposed the QHD method via quantizing the Bregman-Lagrangian framework \cite{Wibisono2016AVP}.
The resulting dynamics can be viewed as the movement of a quantum particle with wave function $\Psi(t)$ satisfying
\begin{equation}\label{eqn:QHD-equation}
\frac{\d \Psi(t)}{\d t} = -i H(t)\Psi(t),
\end{equation}
where
\begin{equation}\label{eqn:QHD-Hamiltonian}
H(t) = e^{\psi(t)}\left(-\frac{1}{2}\Delta\right)+e^{\chi(t)}\hat{f}.
\end{equation}
The corresponding evolution channel for time period $[0,T]$ is $\Lambda_{QHD}(0,T)$. The dynamical system for QHD is guaranteed to converge on convex objective functions. A more rigorous analysis that connects the dynamical system to the discrete QHD algorithm can be found in \cite{chakrabarti2025speedupsconvexoptimizationquantum}.

\subsection{Stochastic Modified Equation and Weak Approximation}
We follow the definition of weak approximation from \cite{Li2017StochasticME,weak-approx-app}, but we consider bounded functions instead.
\begin{definition}[Weak Approximation]
\label{def:weak-approx}
Let $0<\eta<1$, $T>0$ and set $N=\lfloor T/\eta\rfloor$. We say that a continuous stochastic process $X_t,t\in[0,T]$ is an \textit{order-$\alpha$ weak approximation} to a discrete stochastic process $x_k,k=0,1,\ldots,N$ if for every bounded function $g$, there exists $C>0$, independent of $\eta$, such that for all $k=0,1,\dots,N$, 
\begin{equation}
| \E [g(X_{k\eta})]-\E [g(x_{k})] | < C \eta^\alpha.\nonumber
\end{equation}
\end{definition}

\begin{figure}[htbp]
    \centering
    \includegraphics[width=\linewidth]{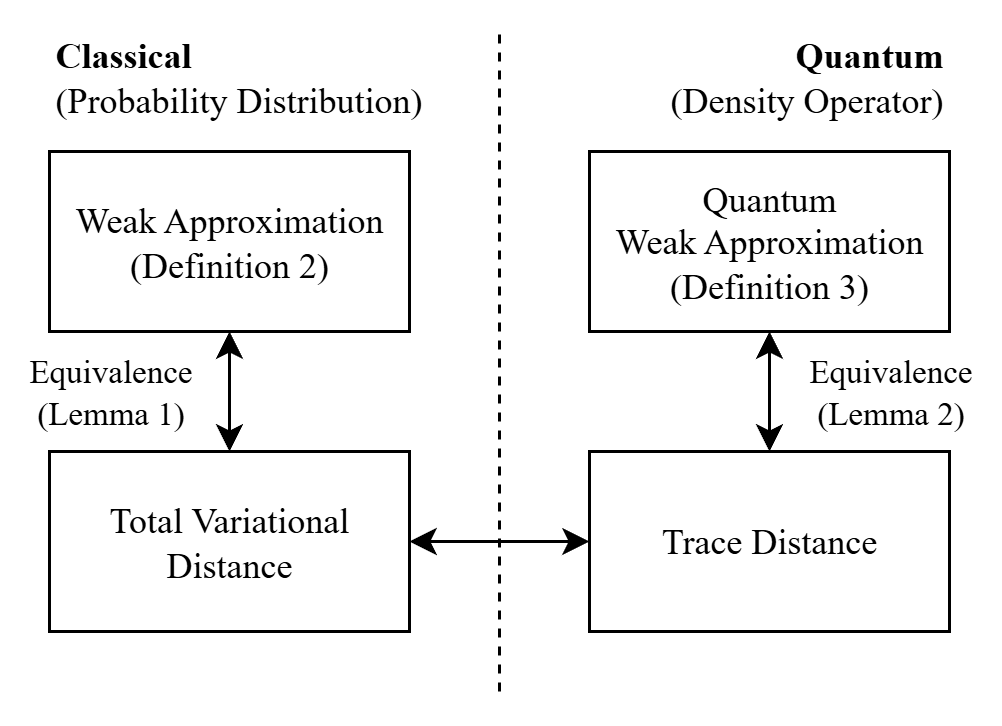}
    \caption{Conceptual schema illustrating the relationship between weak approximations and distance measures in different contexts. Weak approximation, which measures the similarity between sequences of probability distributions, provides a formal connection between optimization algorithms and their continuous-time limits. We extend this measure of approximation to the quantum regime, where density operators (instead of probability distributions) are considered.}
    \label{fig:weak-approx-schema}
\end{figure}

The weak approximation property can be induced from the total variation distance between the probability distributions of $X_{k\eta},x_k$.
\begin{lemma}
\label{lem:weak-approx-equivalence}
Consider a continuous stochastic process $x_C=\{X_t|t\in[0,T]\}$ and a discrete stochastic process $x_D=\{x_k|k=0,1,\ldots,N\}$, and their probability density functions are respectively $P(X_t)$ and $P(x_k)$. If there exists $C'>0$, independent of $\eta$, such that for all $k=0,1,\dots,N$, 
\begin{align*}
d_{TV}(X_{k\eta},x_k)=\frac{1}{2}\int |P(X_{k\eta})[x]-P(x_k)[x]|dx
< C'\eta^\alpha,
\end{align*}
then $x_C$ is an order-$\alpha$ weak approximation for $x_D$.
\end{lemma}
\begin{proof}
For any bounded function $g$, there exists constant $G>0$ such that $|g(x)|<G,\forall x$. Then for any random variable $A$ and $B$,
\begin{align*}
\E [g(A)]&=\int P(A)[x] g(x)\d x, \\
|\E [g(A)]-\E [g(B)]|
&=\left|\int (P(A)[x]-P(B)[x]) g(x)\d x\right| \\
&\leq\int |P(A)[x]-P(B)[x]| |g(x)|\d x \\
&\leq d_{TV}(A,B)G.
\end{align*}
Let $A=X_{k\eta}, B=x_k$, then for $k=0,1,\ldots,N$
\begin{align*}
| \E [g(X_{k\eta})]-\E [g(x_{k})] |\leq d_{TV}(X_{k\eta},x_k)G<C'G\eta^\alpha,
\end{align*}
and $x_C$ is an order-$\alpha$ weak approximation for $x_D$.
\end{proof}

The idea of weak approximation can be applied to the quantum regime, where density operators (instead of probability distributions) are considered:
\begin{definition}[Quantum Weak Approximation]
\label{def:quantum-weak-approx}
Let $0<\eta<1$, $T>0$ and set $N=\lfloor T/\eta\rfloor$. We say that a continuous quantum process $\tilde\rho_t,t\in[0,T],\tilde\rho_t\in L(\mathcal{H})$ is an {\it order-$\alpha$ quantum weak approximation} to a discrete quantum process $\rho_k,k=0,1,\ldots,N,\rho_k\in L(\mathcal{H})$ if for any observable $O$ i.e. linear bounded hermitian operator on $\mH$, there exists $C>0$, independent of $\eta$, such that for all $k=0,1,\dots,N$, 
	\begin{align*}
	\vert \tr(O\tilde\rho_{k\eta})-\tr(O\rho_{k}) \vert < C \eta^\alpha.\nonumber
	\end{align*}
\end{definition}
The quantum weak approximation property can be induced from the trace distance between the density operators.
\begin{lemma}\label{lem:td}
Consider a continuous quantum process $x_C=\{\tilde\rho_t|t\in[0,T]\}$ and a discrete quantum process $x_D=\{\rho_k|k=0,1,\ldots,N\}$. If there exists $C'>0$, independent of $\eta$, such that for all $k=0,1,\dots,N$, the trace distance between $\tilde\rho_{k\eta}$ and $\rho_k$ satisfy
\begin{align*}
\frac{1}{2}\tracenorm{\tilde\rho_{k\eta} - \rho_k}
< C'\eta^\alpha,
\end{align*}
where $\tracenorm{\cdot}$ denotes trace norm, then $x_C$ is an order-$\alpha$ quantum weak approximation for $x_D$.
\end{lemma}
\begin{proof}
For any observable $O$ and $k=0,1,\ldots,N$,
\begin{align*}
\vert \tr(O\tilde\rho_{k\eta})-\tr(O\rho_{k}) \vert
&= \tracenorm{O(\tilde\rho_{k\eta}-\rho_{k})} \nonumber \\
&\leq \opnorm{O} \tracenorm{\tilde\rho_{k\eta} - \rho_k}\nonumber \\
&< 2\opnorm{O}C'\eta^\alpha, \nonumber
\end{align*}
where $\opnorm{\cdot}$ denotes spectral norm for an operator, and $C=2\opnorm{O} C'$ is a constant independent of $\eta$. Therefore, $x_C$ is an order-$\alpha$ quantum weak approximation for $x_D$.
\end{proof}
We can express a random variable $X$ with probability density function $p_X$ as a quantum state with density operator $\rho_X=\int p_X(x)\ket{x}\bra{x}\d x$. In this form, the trace distance between $\tilde\rho_{k\eta}$ and $\rho_k$ coincides with the total variation distance $d_{TV}(X,\tilde X)$.

\section{Discussion}
In this work, we have introduced and analyzed the SQHD method. Through a carefully constructed open quantum system model, we established a rigorous convergence guarantee for convex objectives and demonstrated that our discrete-time algorithm faithfully approximates the continuous dynamics. Beyond theoretical insights, our numerical results confirm SQHD's capability to navigate complex, non-convex landscapes. While these findings position SQHD as a promising candidate for scalable quantum optimization, there are also several caveats and promising directions for future research.

\paragraph{Trade-off Between Cost and Convergence} A primary consideration for SQHD is the trade-off between the computational cost per iteration and the total number of iterations required for convergence. By design, SQHD has a lower cost for querying the objective functions $\{f_1,\ldots,f_m\}$ within a single iteration when compared to the standard QHD algorithm. However, our preliminary observations suggest that the convergence rate of SQHD is generally slower than that of QHD. The specific conditions and problem classes where the reduced iteration cost of SQHD outweighs its slower convergence to provide an overall speedup remain an open and important question.

\paragraph{Hyperparameter Selection} The performance of both QHD and SQHD is sensitive to the choice of hyperparameters. The authors of the original QHD method \cite{QHD} suggest that its effectiveness, particularly on non-convex problems, stems from a "global search" phase inherent in the optimization dynamics. We believe a similar phase exists for SQHD; however, excessive randomness can weaken its effectiveness. This stochasticity, which is governed by hyperparameters such as the learning rate and the Hamiltonian coefficients, must be carefully managed. A rigorous, quantitative analysis of how these hyperparameters influence the optimization landscape and the global search phase would be invaluable. Such a study would not only deepen our theoretical understanding but also provide practical guidance for applying SQHD to real-world tasks.

\paragraph{Future Directions} Based on our findings, we identify several exciting avenues for future work.
First, the framework we have developed could be extended to more recent variants of quantum optimization algorithms based on dynamical systems, such as the gradient-based QHD \cite{gradient-QHD} and Quantum Langevin Dynamics for optimization \cite{QLD}.
Second, we plan to conduct a more comprehensive exploration of the optimal operating regimes for SQHD. This involves characterizing the problem structures and hyperparameter ranges where the advantages of the stochastic approach are most prominent.
Finally, we look forward to running SQHD on current quantum computing devices and applying it to practical, real-world optimization problems to benchmark its performance and demonstrate its utility beyond theoretical analysis.

\section*{Acknowledgements}
This work was supported in part by the National Natural Science Foundation of China Grants No. 92465202, 62325210, 12447107, 62272441, 12204489, 62301531, and Beijing Natural Science Foundation No. 4252013, 4252012.

%\clearpage
\bibliography{sample.bib}
%\clearpage

\onecolumngrid
\appendix
\section{Problem Setting and Assumptions}
\label{sec:setting-and-assumption}
We consider the $d$-dimensional unconstrained finite-sum optimization problem on the domain $\mC=[-1,1]^d\subset \mathbb{R}^d$, expressed as
\begin{align}
\min_{x\in\mC} f(x),f(x)=\frac{1}{m}\sum_{j=1}^m f_j(x), f_j:\mC\to\mathbb{R}.
\end{align}
Only query on $f_j(x),j=1,\ldots,m$ (instead of the total objective function $f(x)$) is allowed.
Consider the Hilbert space $\mH$ on $\mC$. Then $\hat f=\int_{x\in \mC}f(x)\ket{x}\bra{x}\d x$ is a linear operator on $\mH$.
For operators, we use $\opnorm{\cdot}$ for operator norm, $\tracenorm{\cdot}$ for trace norm.
For vectors, we use $\opnorm{\cdot}$ for $\ell^2$ norm.

We also consider the following definition of smoothness on functions, quantum states, mixed quantum states, and operators.
\begin{definition}\label{def:smooth}
Let $x =[x_1,\ldots, x_d]$. A function $f:\mC\to\mathbb{C}$ is smooth up to order $K$ with bound $C$ if for $k=0,1,\ldots,K$ and $j_l\in\{1,\ldots,d\}$, $\prod_{l=1}^{k}\nabla_{j_l} f=\frac{\partial^k f}{\partial x_{j_1}\partial x_{j_2}\cdots \partial x_{j_k}}$ exists and
\begin{align}
\max_{x\in \mC}\left|\prod_{l=1}^{k}\nabla_{j_l} f(x)\right|\leq C.
\end{align}
Let $\nabla_j=\int_{x} \frac{\partial}{\partial x_j}\ket{x}\bra{x}\d x$. A quantum state $\ket{\psi}=\int_x \psi(x)\ket{x}\d x$ on $\mH$ is smooth up to order $K$ with bound $C$ if for $k=0,1,\ldots,K$ and $j_l\in\{1,\ldots,d\}$, $\prod_{l=1}^{k}\nabla_{j_l}\ket{\psi}=\int_{x} \frac{\partial^k \psi( x)}{\partial x_{j_1}\partial x_{j_2}\cdots \d x_{j_k}}\ket{x}\d x$ exists and
\begin{align}
\opnorm{\prod_{l=1}^{k}\nabla_{j_l}\ket{\psi}}\leq C.
\end{align}
A mixed quantum state $\rho$ is smooth up to order $K$ with bound $C$ if for $k=0,1,\ldots,K$ and $j_l\in\{1,\ldots,d\}$, $\prod_{l=1}^{k}\nabla_{j_l} \rho$ exists and
\begin{align}
\tracenorm{\prod_{l=1}^{k}\nabla_{j_l} \rho}\leq C.
\end{align}
An operator $\sigma$ on $\mH$ is smooth up to order $K$ with bound $C$ if for $k=0,1,\ldots,K$, $k_1=0,\ldots,k$, $k_2=k-k_1$ and $j_l\in\{1,\ldots,d\}$, $\sigma'=\left(\prod_{l=1}^{k_1}\nabla_{j_l}\right) \sigma \left(\prod_{l=k_1+1}^{k_1+k_2}\nabla_{j_l}\right)$ exists and
\begin{align}
\tracenorm{\sigma'}\leq C\tracenorm{\sigma}.
\end{align}
\end{definition}

\section{Stochastic Quantum Hybrid Dynamics is a valid Lindblad Master Equation}
Let $A_0=u(t)H$, $A_j=u(t)\sqrt{\eta} e^{\chi(t)}\hat{f}_j,j=1,\ldots,m$, and
\begin{align}
\gamma_{jk}=\delta_{jk}\frac{1}{m}-\frac{1}{m^2}, j,k=1,\ldots,m.
\end{align}
The matrix $\gamma$ is positive semidefinite because for any $x\in\mathbb{C}^m$
\begin{align}
x\gamma x^\dagger
&=\sum_{j,k=1}^m \gamma_{jk}x_j \bar{x}_k \\
&= \frac{\sum_{j=1}^m{|x_j|^2}}{m}-\left|\frac{\sum_{j=1}^m x_j}{m}\right|^2\geq 0.
\end{align}
Then Eq.\eqref{eqn:LME-for-SQHD} is of the form
\begin{align}
&\frac{\d \sigma(t)}{\d t} = -i[A_0,\sigma] + \sum_{j=1}^m\sum_{k=1}^m\gamma_{jk}\left(A_j\sigma A_k^\dagger -\frac{1}{2}(A_k^\dagger A_j\sigma+\sigma A_k^\dagger A_j)\right), 
\end{align}
where  $A_0$ is Hermitian and $\gamma$ is positive semidefinite. This means that Eq.\eqref{eqn:LME-for-SQHD} is a valid Lindblad master equation.

\section{Formal Statement and Proof of the Convergence Result}\label{sec:proof-conv}
The main convergence theorem is as follows:
\begin{theorem}
\label{thm:order-2-conv}
Assume that $f=\frac{1}{m}\sum_{j=1}^m f_j$ is the sum of functions $\{f_1,\ldots,f_m\}$. Assume $f_j$ is a convex function smooth up to order $2$ with bound $ L_j$ for $j=1,\ldots,m$.
Let $\sigma_f^*=\E_{j}[\opnorm{\nabla f_j (x^*)-\nabla f(x^*)}^2]$ be the gradient noise of the function $f$.
Let $L_{\max}=\max_j L_j$.
Let $x^*$ be the unique local minimizer of $f$.

For any initial state $\rho_0$ smooth up to order 2, and functions $\alpha,\beta,\gamma:[0,T]\to\mathbb{R}$ that satisfies the strong ideal scaling condition
\begin{align}
\psi(t) &= \alpha(t)-\gamma(t), \\
\chi(t) &= \alpha(t)+\beta(t)+\gamma(t), \\
\dot{\beta}(t) &= \dot{\gamma}(t) = e^{\alpha(t)}.
\end{align}
the result state $\rho_t$ of Stochastic Quantum Hybrid Dynamics \eqref{eqn:LME-for-SQHD} with learning rate schedule $u(t)=e^{-(\alpha(t)+\beta(t))}$ satisfies
\begin{align}
\braket{\hat{f}}_t - f(x^*) \leq C_1 e^{-\beta(\tau)}+C_2 \eta(\sigma_f^*+dL_{\max}^2),
\end{align}
for some constant $C_1,C_2>0$, where $\tau = \int_0^t u(s)\d s$ is the effective time.
\end{theorem}

Before we prove Theorem \ref{thm:order-2-conv}, we need the following lemmas.
\begin{lemma}
Let $\rho_t$ be the solution to a Lindblad master equation
\begin{align}
\frac{d\rho_t}{d t}=\mL(t)[\rho_t]=-i[H(t),\rho_t]+\sum_{j,k} a_{jk} (A_j\rho_t A_k^\dagger-\frac{1}{2}\{A_k^\dagger A_j,\rho_t\}),
\end{align}
with initial state $\rho_0$. The expectation of observable is defined as $\braket{O}_t=\tr(O\rho_t)$. Then
\begin{align}
\frac{d\braket{O(t)}_t}{d t}
&=\braket{\frac{d O(t)}{d t}}_t+i\braket{[H(t),O(t)]}_t+\sum_{j,k} a_{jk} \braket{(A_k^\dagger O(t) A_j-\frac{1}{2}\{A_k^\dagger A_j,O(t)\})}_t.
\end{align}
\end{lemma}
\begin{proof}
\begin{align}
\frac{d\braket{O(t)}_t}{d t}
&= \frac{\d}{\d t}\tr(\rho_t O(t)) \\
&= \tr(\frac{\d \rho_t}{\d t}O(t)) + \tr(\rho_t\frac{\d O(t)}{\d t}) \\
&= \tr(\mL(t)[\rho_t]O(t)) + \braket{\frac{\d O(t)}{\d t}}_t. \label{eqn:appendix-A-lemma-1-first-decomposition}
\end{align}
Breaking down $\tr(\mL(t)[\rho_t]O(t))$, we have
\begin{align}
\tr(-i[H(t),\rho_t] O(t))
&= -i \tr(O(t)H(t)\rho_t-H(t)O(t)\rho_t) \\
&= \braket{i[H(t),O(t)]}_t, \\
\tr((A_j\rho_t A_k^\dagger-\frac{1}{2}\{A_k^\dagger A_j,\rho_t\})O(t))
&=  \tr(A_k^\dagger O(t) A_j \rho_t -\frac{1}{2}O(t)A_k^\dagger A_j\rho_t-\frac{1}{2}A_k^\dagger A_j O(t)\rho(t)) \\
&= \braket{A_k^\dagger O(t) A_j-\frac{1}{2}\{A_k^\dagger A_j,O(t)\})}_t.
\end{align}
Plug in Eq.\eqref{eqn:appendix-A-lemma-1-first-decomposition} and the proof is finished.
\end{proof}
We let
\begin{align}
\hat{p}_j&=-i\nabla_j, \\
\hat{p}&=-i\nabla=[-i\nabla_1,\ldots,-i\nabla_d], \\
\hat{x}_j &= \int_{x} x_j \ket{x}\bra{x}\d x, \\
\hat{x} &= [\hat{x}_1,\ldots,\hat{x}_d].
\end{align}
Then $\Delta=-\sum_{j=1}^d \hat{p}_j^2=\nabla\cdot \nabla$.
\begin{lemma}
% \begin{align}
% [\Delta,\hat{x}\cdot\hat{x}] &= 4i{x}\cdot{\hat{p}} + 2,
% &[f,\hat{x}\cdot\hat{x}] &= 0,\\
% [\Delta, \hat{x}\cdot\hat{p}] &= -2i\Delta,
% &[f, \hat{x}\cdot\hat{p}] &= i{x}\cdot \widehat{\nabla f},\\
% [\Delta, \Delta] &= 0,
% &[f, \Delta] &= \Delta f+2(\widehat{\nabla f})\cdot\nabla,\\
% [\Delta, f] &= -[f,\Delta],
% &[f, f] &= 0,
% \end{align}
\begin{align}
[\hat{f}, \hat{x}\cdot\hat{p}] &= i\hat{x}\cdot \widehat{\nabla f},
&[\hat{f}, \hat{p}\cdot\hat{x}] &= i\hat{x}\cdot \widehat{\nabla f},\\
[\hat{f},\hat{x}\cdot\hat{x}] &= 0,
&[\hat{f}, \hat{p}\cdot\hat{p}] &= \Delta f+2(\widehat{\nabla f})\cdot\nabla.
\end{align}
\end{lemma}
\begin{proof}
Consider a smooth test function $\ket{\psi}$.

$[\hat{f},\hat{x}_k \hat{p}_k]\ket{\psi}=-i\int_x \left( f(x)x_k \frac{\partial \psi}{\partial x_k} - x_k\frac{\partial}{\partial x_k}\left(f(x)\psi(x)\right)\right)\ket{x}\d x$, then $[\hat{f},\hat{x}\cdot\hat{p}]=\sum_{k=1}^d [\hat{f},\hat{x}_k \hat{p}_k]= i \hat{x}\cdot \widehat{\nabla f}$.

$[\hat{f},\hat{p}_k \hat{x}_k]\ket{\psi}=-i\int_x \left( f(x)\frac{\partial }{\partial x_k}\left(x_k\psi(x)\right) - \frac{\partial}{\partial x_k}\left(x_kf(x)\psi(x)\right)\right)\ket{x}\d x$, then $[\hat{f},\hat{p}\cdot\hat{x}]=\sum_{k=1}^d [\hat{f},\hat{p}_k \hat{x}_k]= i \hat{x}\cdot \widehat{\nabla f}$.

$\hat{f}$ and $\hat{x}_j^2$ commutes, therefore $[\hat{f},\hat{x}\cdot\hat{x}]=0$.

$[\hat{f},\hat{p}_k^2]\ket{\psi}=-\int_x \left(f(x) \frac{\partial^2 \psi}{\partial x_j^2} - \frac{\partial^2}{\partial x_j^2}\left(f(x)\psi(x)\right)\right)\ket{x}\d x$, and
\begin{align}    
[\hat{f},\hat{p}\cdot\hat{p}] = \sum_k [\hat{f},\hat{p}_k^2] = \Delta f + 2(\widehat{\nabla f})\cdot \nabla.
\end{align}
\end{proof}

\begin{lemma}
\label{lem:var-trans}
Suppose $f_j$ is convex and smooth up to order $2$ with bound $L_j$ for $j=1,\ldots,m$. Let $L_{\max}=\max_j L_j$. Then for all $x\in\mathbb{R}^d$
\begin{align}
\E_j[\opnorm{\nabla f_j(x)-\nabla f(x)}^2]
&= \frac{1}{m}\sum_{j=1}^m \opnorm{\nabla f_j(x)}^2 - \opnorm{\nabla f(x)}^2\\
&\leq  8dL_{\max}^2+\sigma_f^*.
\end{align}
\end{lemma}
\begin{proof}
To start with, for all $x\in \mathbb{R}^d$
\begin{align}
&\E_j[\opnorm{\nabla f_j (x) - \nabla f(x)}^2] \\
\leq& \E_{j}[\opnorm{\nabla f_j (x)-\nabla f_j (x^*)}^2] \\
+& \E_{j}[\opnorm{\nabla f_j (x^*)-\nabla f (x^*)}^2] \\
+& \E_{j}[\opnorm{\nabla f (x^*)-\nabla f (x)}^2].
\end{align}
The first term and the third term are related to the expected smoothness of functions $f_j$ and $f$. The second term is the gradient noise $\sigma_f^*$ defined in Theorem \ref{thm:order-2-conv}.
For a function $g$ smooth up to order $2$ with bound $B$,
\begin{align}
|\nabla_j g(x)-\nabla_j g(y)|
&\leq \sum_{k=1}^d \max_{z} \left|\frac{\partial^2 g}{\partial z_k\partial z_j}\right| |x_k-y_k| \\
&\leq \sqrt{d}B \opnorm{x-y}, \\
\opnorm{\nabla g(x)-\nabla g(y)} &\leq \sqrt{\sum_{j=1}^d |\nabla_j g(x)-\nabla_j g(y)|^2} \\
&\leq dB \opnorm{x-y}.
\end{align}
Since $f_j$ is convex and smooth up to order $2$ with bound $L_j$ and $f$ is convex and smooth up to order $2$ with bound $L_{\max}$, then
\begin{align}
\opnorm{\nabla f_j(x)-\nabla f_j(y)}&\leq dL_j \opnorm{x-y}, \\
\opnorm{\nabla f(x)-\nabla f(y)}&\leq dL_{\max} \opnorm{x-y}.
\end{align}
by \cite[Lemma 2.29]{handbook} we know
\begin{align}
\opnorm{\nabla f_j (x)-\nabla f_j (x^*)}^2&\leq 2dL_j \left(f_j(x)-f_j(x^*)\right), \\
\opnorm{\nabla f (x)-\nabla f (x^*)}^2&\leq 2dL_{\max} \left(f(x)-f(x^*)\right).
\end{align}
Then the first term and the third term are both upper bounded by
$2dL_{\max}(f(x)-f(x^*))$, and $|f(x)|\leq \frac{1}{m}\sum_{j=1}^m L_j\leq L_{\max}$. Then
\begin{align}
\E_j[\opnorm{\nabla f_j(x)-\nabla f}(x)_2^2]\leq  8dL_{\max}^2+\sigma_f^*.
\end{align}
\end{proof}

Below is the proof for Theorem \ref{thm:order-2-conv}.
\begin{proof}
Since $\tau(t) = \int_0^t u(s)\d s$, then
\begin{align}
\frac{\d \rho}{\d t} &= \frac{\d\rho}{\d \tau}\frac{\d \tau}{\d t}, \\
\frac{\d \tau}{\d t} &= u(t), \tau(0) = 0, \\
\frac{\d\rho}{\d \tau} &= \mL_{GD}[\rho]+u(\tau)\eta\mL_{NOISE}[\rho].
\end{align}
We assume $x^*=0,f(x^*)=0$ without loss of generality.
Consider the following Lyapunov function \cite[Theorem 1]{QHD}:
\begin{align}
\hat {E}(\tau) &= (\hat{x}+e^{-\gamma(\tau)}\hat{p})^2/2+e^{\beta(\tau)}\hat{f}, \\
E(\tau) &= \braket{\hat{E}(\tau)}_\tau.
\end{align}
$\hat f$ is a bounded operator, and $\rho_0$ is smooth up to order $2$, therefore $E(0)$ is bounded. Next we provide an upper bound of $E(\tau)$'s derivative, and therefore provide an upper bound of $E(\tau)$ and $\braket{\hat f}_\tau$.
\begin{align}
\frac{\d E(\tau)}{\d \tau}
=&  \left(\braket{\frac{\d \hat{E}(\tau)}{\d \tau}}_\tau+i\braket{[H(\tau),\hat{E}(\tau)]}_\tau\right) \\
&+{u(\tau)\eta e^{2\chi(\tau)}}\left(\frac{1}{m}\sum_{j}  \braket{\hat{f}_j \hat{E}(\tau) \hat{f}_j-\frac{1}{2}\{\hat{f}_j^2,\hat{E}(\tau)\}}_\tau - \braket{\hat{f} \hat{E}(\tau) \hat{f}-\frac{1}{2}\{\hat{f}^2,\hat{E}(\tau)\}}_\tau\right) ,\label{eqn:dEdt}
\end{align}
The first term in Eq.\eqref{eqn:dEdt} is bounded due to  \cite[Proposition 2 ]{QHD}:
\begin{align}
\braket{\frac{d \hat{E}(\tau)}{d \tau}}_\tau+i\braket{[H(\tau),\hat{E}(\tau)]}_\tau\leq e^{\alpha(\tau)+\beta(\tau)}\braket{\hat{f}-\hat{x}\cdot \widehat{\nabla f}}_\tau.
\end{align}
The second term in Eq.\eqref{eqn:dEdt} becomes
\begin{align}
e^{2(\alpha(\tau)+\beta(\tau))}u(\tau)\frac{\eta}{2}\left(\frac{1}{m}\sum_{j=1}^m  \braket{\|\widehat{\nabla f_j}\|^2}_\tau
-
\braket{\|\widehat{\nabla f}\|^2}_\tau
\right),
\end{align}
because
\begin{align}
\hat{g} \hat{E} \hat{g}-\frac{1}{2}\{\hat{g}^2,\hat{E}\} &= \frac{1}{2}[\hat{g},[\hat{E},\hat{g}]], \quad g=f,f_1,\ldots,f_m, \\
[\hat{g},[\hat{f},\hat{g}]] & = 0, \\
[\hat{g},[\hat{x}\cdot\hat{x},\hat{g}]] & = 0, \\
[\hat{g},[\hat{x}\cdot\hat{p},\hat{g}]] & = [\hat{g},i\hat{x}\cdot \widehat{\nabla g}] = 0, \\
[\hat{g},[\hat{p}\cdot\hat{x},\hat{g}]] & = [\hat{g},i\hat{x}\cdot \widehat{\nabla g}] = 0, \\
[\hat{g},[\Delta,\hat{g}]]
& = -[\hat{g},\widehat{\Delta g}+2(\widehat{\nabla g})\cdot\nabla] \\
&=-( \hat{g}\widehat{\Delta g}+2\hat{g}(\widehat{\nabla g})\cdot\nabla
-
((\widehat{\Delta g})\hat{g}+2(\widehat{\nabla g})\cdot(\widehat{\nabla g})+2\hat{g}(\widehat{\nabla g})\cdot\nabla)
)\\
&=2\|\widehat{\nabla g}\|^2,
\end{align}
where $\|\widehat{\nabla g}\|^2=\sum_{j=1}^d (\widehat{\nabla_j g})^2$, and
\begin{align}
e^{2\chi(\tau)}u(\tau)\eta\frac{e^{-2\gamma(\tau)}}{2} = e^{2(\alpha(\tau)+\beta(\tau))}u(t)\frac{\eta}{2} = e^{\alpha(\tau)+\beta(\tau)}\frac{\eta}{2}. 
\end{align}
Therefore
\begin{align}
\frac{\d E(\tau)}{\d \tau}
&\leq
e^{\alpha(\tau)+\beta(\tau)}\left(
\braket{\hat{f}-\hat{x}\cdot \widehat{\nabla f}}_\tau
+
\frac{\eta}{2} \left(\frac{1}{m}\sum_{j=1}^m  \braket{\|\widehat{\nabla f_j}\|^2}_\tau
-
\braket{\|\widehat{\nabla f}\|^2}_\tau
\right)
\right).
\end{align}
Since $f$ is smooth up to order $2$ and convex, %\sirui{How to prove that $G$ exists? One (not so elegant) way is to limit the variable domain of $f$ to a closed region $M$, and hence $f$ is bounded.}
\begin{align}
\braket{\hat{f}-\hat{x}\cdot \widehat{\nabla f}}_t &\leq 0, \\
\frac{1}{m}\sum_{j=1}^m  \braket{\|\widehat{\nabla f_j}\|^2}_\tau
-\braket{\|\widehat{\nabla f}\|^2}_\tau
&\leq 4dL_{\max}\braket{\hat f}_\tau+\sigma_f^*\leq 8dL_{\max}^2 + \sigma_f^*.
\end{align}
Then
\begin{align}
E(\tau)-E(0) &\leq \int_0^\tau e^{\alpha(s)+\beta(s)}\frac{\eta }{2} (4dL_{\max}G+\sigma_f^*)\d s \\
&= (e^{\beta(\tau)}-e^{\beta(0)})\frac{\eta }{2}(4dL_{\max}G+\sigma_f^*) \\
&\leq e^{\beta(\tau)}\frac{\eta }{2}(4dL_{\max}G+\sigma_f^*),
\end{align}
because $\frac{\d}{\d t}(e^{\beta(t)})=\dot\beta(t) e^{\beta(t)}=e^{\alpha(t)+\beta(t)}$. Then
\begin{align}
\braket{\hat f}_\tau &\leq e^{-\beta(\tau)}E(\tau) \\
&\leq e^{-\beta(\tau)}E(0) + \frac{\eta }{2} (4dL_{\max}G+\sigma_f^*), \\
&\leq C_1 e^{-\beta(\tau)}+ C_2 \eta(dL_{\max}G+\sigma_f^* ),
\end{align}
where
\begin{align}
C_1 &= E(0), \\
C_2 &= 2, \\
\tau(t)&=\int_0^t u(s)\d s.
\end{align}
\end{proof}

\section{Formal Statement and Proof of the Approximation Result}\label{sec:proof-approx}
We start with a rigorous definition of the SQHD algorithm.
\begin{definition}[Stochastic Quantum Hamiltonian Descent]
Stochastic Quantum Hamiltonian Descent with iteration number $N$ and learning rate $\eta$ is defined as
\begin{align}
&U_{dSQHD}(N,\eta;\xi) \\
=&\prod_{j=N-1}^{0}\left[ \exp(-i \frac{\eta}{2} a_j (-\Delta/2)) \exp(-i\eta b_j \hat{f}_{\xi_j})\exp(-i \frac{\eta}{2} a_j (-\Delta/2))\right],
\end{align}
where $\xi$ is a $N$-dimensional random vector that $\xi_j$ is independently and uniformly drawn from $\{1,\ldots,m\}$ for $j=0,\ldots,N-1$. The corresponding channel is denoted as
\begin{align}
\Lambda_{dSQHD}(N,\eta)[\rho] = \E_{\xi}\left[U_{dSQHD}(N,\eta;\xi)\rho U^\dagger_{dSQHD}(N,\eta;\xi)\right].
\end{align}
\end{definition}

The convergence result (Theorem \ref{thm:order-2-conv}) is about $\Lambda_{LS}(0,T)[\rho_0]$. However, such a process can not be directly run on a digital quantum computer. Instead, we implement a gate-based quantum algorithm. We need an additional approximation result to argue the convergence of the discrete quantum algorithm.
The key obstacle is the unbounded Hamiltonian, which forbids approximation analysis regarding the operator norm. As a fallback, we consider the approximation error on the vector norm \cite{An2021TimeDU}. The intuition is that while the operator norm of the Hamiltonian is unbounded, the vector norm for the result state can be bounded. Therefore, we expect the gradient norm $\Delta\psi$ to be bounded, which leads to Definition \ref{def:smooth}. The following approximation theorem is based on the smoothness assumption on the quantum states:
\begin{theorem}
\label{thm:order-2-approx}
Given iteration number $N$ and learning rate $\eta\in(0,1)$.
For any initial state $\rho_0$ such that $\tilde \rho_t=\Lambda_{LS}(0,t)[\rho_0]$ is smooth up to order $6$ for $t\in[0,N\eta]$, any functions $e^{\psi(t)},e^{\chi(t)}$ smooth up to order $3$ and $f_j(x),j=1,\ldots,m$ smooth up to order $6$, the Stochastic Quantum Hybrid Dynamics process $\tilde \rho_t$ is an order-$2$ quantum weak approximation of the Stochastic Quantum Hamiltonian Descent process $\rho_k=\Lambda_{dSQHD} (k,\eta)[\rho_0]$.
\end{theorem}
Theorem \ref{thm:order-2-approx} (in the broadest sense) is a quantization of \cite[Theorem 1]{Li2017StochasticME}.
Before we present the proof of Theorem \ref{thm:order-2-approx}, we state a few useful lemmas and prove them.  

\begin{lemma}
If $\{\ket{\psi_1},\ldots,\ket{\psi_m}\}$ are quantum states smooth up to order $K$, then $\rho=\sum_{j=1}^m p_j \ket{\psi_j}\bra{\psi_j}$ is a mixed quantum state smooth up to order $K$.
\end{lemma}
\begin{proof}
To start with, for any operator $A$,$\tracenorm{A\ket{\psi}\bra{\psi}}=\opnorm{A\ket{\psi}}\cdot\opnorm{\ket{\psi}}$. For $A\ket{\psi}=0$ both side equal $0$. Else $\operatorname{rank}(A\ket{\psi}\bra{\psi})=1$, and its only non-zero singular value is $\opnorm{A\ket{\psi}}\cdot\opnorm{\ket{\psi}}$.

For $k=0,1,\ldots,K$, $\tracenorm{\prod_{l=1}^{k}\nabla_{j_l} \rho}=\sum_{j=1}^m p_j \tracenorm{\prod_{l=1}^{k}\nabla_{j_l}\ket{\psi_j}\bra{\psi_j}}$, and
\begin{align}
\tracenorm{\prod_{l=1}^{k}\nabla_{j_l}\ket{\psi_j}\bra{\psi_j}}
\leq \opnorm{\prod_{l=1}^k\nabla_{j_l}\ket{\psi_j}}\opnorm{\ket{\psi_j}}
\leq C_j,
\end{align}
where $\ket{\psi_j}$ is smooth up to order $K$ with bound $C_j$. Then $\tracenorm{\prod_{l=1}^{k}\nabla_{j_l} \rho}\leq \sum_j p_j C_j$ for $k=0,1,\ldots,K$, and $\rho$ is smooth up to order $K$ with bound $\sum_j p_j C_j$.
\end{proof}

\begin{lemma}\label{lem:constant-Hamiltonian-smoothness-preservation}
Let $H=\Delta+\hat f$, where $\hat f=\int_{x\in\mathbb{R}^d} f(x)\ket{x}\bra{x}\d x$. $f:\mathbb{R}^d\to\mathbb{R}$ is smooth up to order $n$. There exists a non-zero period vector $T\in\mathbb{R}^d$ such that for all $x\in\mathbb{R}^d$, $f(x)=f(x+T)$. Then for any mixed quantum state $\rho_0$ smooth up to order $n$, $\rho_1=e^{-iH}\rho_0 e^{iH}$ is smooth up to order $n$.
\end{lemma}
The lemma follows directly form \cite[Lemma 6.2]{Bourgain1999}, according to \cite{An2021TimeDU}. The function we consider, namely function $f: \mathcal{C} \to \mathbb{R}$ smooth up to order $n$, satisfies the conditions of Lemma \ref{lem:constant-Hamiltonian-smoothness-preservation}. We can extend the domain of $f$ from $\mathcal{C}$ to $\mathbb{R}^d$ in such a way that the extension is periodic while preserving the smoothness condition.

\begin{lemma}\label{lem:nested-commutator-smoothness-preservation}
Consider the nested commutators
\begin{align}
\mS_n
=\left\{[A_1,[A_2,\cdots[A_{n-1},[A_n,\rho]]\cdots]], A_j\in\{\nabla_1^2,\ldots,\nabla_d^2,\hat f_1,\ldots,\hat f_m\},\sum_{j=1}^n[A_j\in\{\nabla^2_1,\ldots,\nabla^2_d\}]\leq k\right\},
\end{align}
where the operator $\rho$ is smooth up to order $2k$, and operators $\hat g=\int_x g(x)\ket{x}\bra{x}\d x$ are defined for functions $g\in\{f_1,\ldots,f_m\}$ smooth up to order $2k$.

For any $S_n\in\mS_n$, there exists constant $C>0$ such that $\tracenorm{S_n}\leq C\tracenorm{\rho}$.
\end{lemma}
It is worth mentioning that Lemma \ref{lem:nested-commutator-smoothness-preservation} applies not only to mixed quantum states (which are semidefinite positive with trace norm $1$) but also to operators (which may be indefinite with arbitrary trace norm).

\begin{proof}
The target $S_n=[A_1,\cdots[A_n,\rho]\cdots]$ is the linear combination of $O_L\rho O_R$ where $O_L$ and $O_R$ are products of some of the operators $A_k$. 

Consider a generic product $O_L = A_{i_1} A_{i_2} \cdots A_{i_r}$, where $i_1<i_2<\cdots<i_r,1\leq r\leq n$.
We consider the $d=1$ scenario. In this scenario, $O_L$ is the alternation of at most $k$ Laplace operators $\nabla^2$ and composite diagonal operators $\hat g\in\{\hat f_1,\ldots,\hat f_m\}^{r}$. We insert $\hat 1 = \int_x \ket{x}\bra{x}\d x$ between consecutive Laplace operator, then
\begin{align}
O_L=\nabla^{2} \hat g_1 \nabla^{2} \hat g_2 \cdots \nabla^{2} \hat g_k,g_j\in\{\hat 1,\hat f_1,\ldots,\hat f_m\}^{r}.
\label{eqn:appendix-B-interleaved-observable}
\end{align}
Notice that
\begin{align}
\nabla^2 \hat g = \sum_{j=0}^2 {2\choose j}\hat g^{(j)}\nabla^{2-j},
\end{align}
where $\hat g^{(j)}=\widehat{\nabla^j g}=\int_x \left(\frac{\d^j g}{\d x^j}\right)\ket{x}\bra{x}\d x$.
Plug in Eq.\eqref{eqn:appendix-B-interleaved-observable}, we notice that $O_L$ is the linear combination of
\begin{align}
\hat g_1^{(j_1)}\nabla^{4-j_1} \hat g_2 \nabla^{2} \hat g_3 \cdots \nabla^{2} \hat g_k, j_1\leq 2.
\end{align}
We repeat the plugging-in $k$ times, and we know $O_L$ is the linear combination of
\begin{align}
\hat g_1^{(j_1)}\hat g_2^{(j_2)}\cdots \hat g_k^{(j_k)}\nabla^{j_{k+1}},0\leq j_l\leq 2n,\sum_{l=1}^{k+1}j_l= 2n,j_{k+1}\leq 2k.
\label{eqn:appendix-B-lemma-5-pushing-result}
\end{align}
Since functions $\{1,f_1,\ldots,f_m\}$ are smooth up to order $2n$ and thus their composites $g_1,\ldots,g_k$ are also smooth up to order $2n$, we know $G=\hat g_1^{(j_1)}\hat g_2^{(j_2)}\cdots \hat g_k^{(j_k)}$ has bounded operator norm. Therefore, $O_L=\sum_{G_1} G_1\nabla^{k(G_1)}$. Similarly, $O_R=\sum_{G_2} G_2 \nabla^{k(G_2)}$. Notice that there are at most $k$ Laplace operators $\nabla^2$ from $\{A_1,\ldots,A_n\}$, then $k(G_1)+k(G_2)\leq 2k$ for any $G_1,G_2$. By the smoothness of the operator $\rho$, there exists constant $C_1>0$ such that $\tracenorm{\nabla^{k(G_1)}\rho\nabla^{k(G_2)}}\leq C_1\tracenorm{\rho}$.
Therefore, there exists constant $C_2>0$ such that
\begin{align}
\tracenorm{O_L\rho O_R}
&\leq \sum_{G_1,G_2}\tracenorm{G_1 \nabla^{k_1}\rho\nabla^{k_2} G_2} \\
&\leq \sum_{G_1,G_2}\|G_1\| \|\nabla^{k_1}\rho\nabla^{k_2}\|_1 \|G_2\|\leq C_2\tracenorm{\rho}, \\
\end{align}
and constant $C_3>0$ such that
\begin{align}
\tracenorm{S_n}
&\leq \sum_{O_L,O_R} \tracenorm{O_L\rho O_R}\leq C_3\tracenorm{\rho}.
\label{eqn:appendix-B-lemma-5-rest-end}
\end{align}
Therefore, we have proved the lemma when $d=1$.

For the cases of arbitrary $d$, Eq.\eqref{eqn:appendix-B-interleaved-observable} becomes
\begin{align}
& O_L=\left(\nabla^2_{j_1}\right) \hat g_1 \left(\nabla^2_{j_2}\right) \hat g_2 \cdots \left(\nabla^2_{j_k}\right) \hat g_k, \\
& g_j\in\{\hat 1,\hat f_1,\ldots,\hat f_m\}.
\end{align}
and repeating the "pushing" process in that push $\hat g\in\{\hat 1,\hat f_1,\ldots,\hat f_m\}$ to the front leads to the same result in Eq.\eqref{eqn:appendix-B-lemma-5-pushing-result} (except that operator $\hat g^{(j_l)}_l$ replaced by $\hat g^{(j_l)}_l[{\mathrm{idx}_1,\ldots,\mathrm{idx}_{j_l}}]$ for specifying the indices of the derivatives), and the argument from Eq.\eqref{eqn:appendix-B-interleaved-observable} to \eqref{eqn:appendix-B-lemma-5-rest-end} remains unchanged.
Thus we have proved the lemma.
\end{proof}

We note that the Lindbladian $\mL_{LS}(t)$ \eqref{eqn:LME-for-SQHD} can be expressed as a linear combination of time-independent operators with time-dependent coefficients,
\begin{align}
\mL_{NOISE}&=e^{2\chi(t)}\tilde\mL_{NOISE}, \\
\mL_{GD}&=e^{\psi(t)}\tilde\mL_K+e^{\chi(t)}\tilde\mL_P, \\
\tilde\mL_K[\sigma]&=-i\left[-\frac{1}{2}\Delta, \sigma\right], \\
\tilde\mL_P[\sigma]&=-i\left[\hat f, \sigma\right],\tilde\mL_{P,j}[\sigma]=-i\left[\hat f_j, \sigma\right], \\
H_j&= e^{\psi(t)}\left(-\frac{1}{2}\Delta\right)+e^{\chi(t)}\hat f_j,
\end{align}
therefore
\begin{align}
\mL_{LS}(t)=u(t)\left(e^{\psi(t)}\tilde\mL_K+e^{\chi(t)}\tilde\mL_P\right) + u(t)^2\eta e^{2\chi(t)}\tilde\mL_{NOISE}.
\end{align}

\begin{lemma}\label{lem:channel-prod-bounded-derivative}
For any $k\in\mathbb{N}_+$ and $t\in[0,T]$, if $\rho$ is an operator smooth up to order $4k$ and $f_j,j=1,\ldots,m$ are smooth up to order $2k$, then
\begin{align}
\tracenorm{\mB_1\mB_2\cdots \mB_k[\rho]}<\infty,\mB_i\in\{\mI,\tilde\mL_{NOISE},\tilde\mL_K,\tilde\mL_P\}.
\end{align}
\end{lemma}
\begin{proof}
Since
\begin{align}
\tilde\mL_{NOISE}[\sigma]
&=\frac{1}{2}[\hat{f},[\hat{f},\sigma]]-\frac{1}{2m}\sum_{l=1}^m[\hat{f}_l,[\hat{f}_l,\sigma]],\\
\tilde\mL_K[\sigma]
&=\frac{i}{2}\sum_{j=1}^d [\nabla^2_j,\sigma],
\end{align}
$\mB_1\mB_2\cdots \mB_j[\rho]$ again is the linear combination of
\begin{align}
    [B_{2j},[\cdots[B_2,[B_1,\rho]]\cdots]],B_i\in\{\nabla^2_1,\ldots,\nabla_d^2,\hat{f},\hat{f}_1,\ldots,\hat{f}_m\}.
\end{align}
where $\{\nabla^2_1,\ldots,\nabla_d^2\}$ occur in $\{B_1,\ldots,B_{2j}\}$ at most $j\leq k$ times. By Lemma \ref{lem:nested-commutator-smoothness-preservation}, the nested commutator has bounded trace norm.
\end{proof}

With these lemmas, we can prove Theorem \ref{thm:order-2-approx}. In the following, for $\alpha\in\mathbb{N}$ we use $O(\eta^\alpha)$ to denote an operator $\sigma\in \mL(\mH)$ such that $\tracenorm{\sigma}<C\eta^\alpha$ for some constant $C$.

Now we present the proof of Theorem \ref{thm:order-2-approx}.
\begin{proof}
Consider $t_k=k\eta,k=0,\ldots,N-1$, the Taylor series expansion of $\tilde\rho_{t_{k+1}}$ around $t_{k}$ up to the second order is
\begin{align}
\tilde\rho_{t_{k+1}} &= \tilde\rho_{t_k} + \eta \frac{\d \tilde\rho_t}{\d t}\bigg|_{t=t_k} + \frac{\eta^2}{2} \frac{\d^2 \tilde\rho_t}{\d t^2}\bigg|_{t=t_k} +  \frac{\eta^3}{6}\frac{\d^3 \tilde\rho_t}{\d t^3}\bigg|_{t=r_k}\\
\end{align}
where $r_k\in[t_k,t_{k+1}]$.
Let $c_1(t)=u(t)e^{\psi(t)},c_2(t)=u(t)e^{\chi(t)},c_3(t)=u(t)^2e^{2\chi(t)}\eta$, then functions $c_1,c_2,c_3:[0,T]\to\mathbb{R}$ are all smooth up to order $3$. Let
\begin{align}
\mA_1 &= \mL_{LS} = c_1(t)\tilde\mL_K+c_2(t)\tilde\mL_P+c_3(t)\tilde\mL_{NOISE},\\
\mA_{j+1}&=\left(\frac{\d^j c_1}{\d t^j}\tilde\mL_K+\frac{\d^j c_2}{\d t^j}\tilde\mL_P+\frac{\d^j c_3}{\d t^j}\tilde\mL_{NOISE}\right) \\
&+ \mA_j\left(c_1(t)\tilde\mL_K+c_2(t)\tilde\mL_P+c_3(t)\tilde\mL_{NOISE}\right), j\geq 0.
\end{align}
then $\frac{\d^3 \tilde\rho_t}{\d t^3}\bigg|_{t=0} = \mA_3[\rho]$. $\mA_3$ is the linear combination of $\mB_1\mB_2\mB_3,\mB_j\in\{\mI,\tilde\mL_K,\tilde\mL_P,\tilde\mL_{NOISE}\}$. By Lemma \ref{lem:channel-prod-bounded-derivative}, their trace norm is bounded due to the smoothness condition of $\tilde\rho_t$ and $f_1,\ldots,f_m$. The corresponding coefficient is the linear combination of $c_1,c_2,c_3$ and their derivative up to order $3$; thus, these coefficients are all bounded due to the smoothness condition of $e^{\psi(t)},e^{\chi(t)}$. Therefore, $\frac{\d^3 \tilde\rho_t}{\d t^3}\bigg|_{t=r_k}=O(1)$ and
\begin{align}
\tilde\rho_{t_{k+1}}&=\tilde\rho_{t_k}
 + \eta \mathcal{L}_{LS}[\tilde\rho_{t_k}]
 + \frac{\eta^2}{2}\left(
    \dot{\mathcal{L}}_{LS}[\tilde\rho_{t_k}]
    + \mathcal{L}_{LS}^2[\tilde\rho_{t_k}]
    \right)
 + O(\eta^3).
\end{align}
Notice that $\mL_{LS}=\mL_{GD}+\eta\mL_{NOISE}$, $O(\eta^3)$ parts can be ignored in the second-order Taylor series expansion. Therefore
\begin{align}\label{eqn:ts-expansion-LS}
\tilde\rho_{t_{k+1}} &=
\left(\mI
+ \eta \mathcal{L}_{LS}
+ \frac{\eta^2}{2}\dot{\mathcal{L}}_{GD}
+ \frac{\eta^2}{2}\mathcal{L}_{GD}^2
\right)[\tilde\rho_{t_k}]
 + O(\eta^3) \\
 &= \left(\mI
+ \eta \mathcal{L}_{GD}
+ \eta^2 \mathcal{L}_{NOISE}
+ \frac{\eta^2}{2}\dot{\mathcal{L}}_{GD}
+ \frac{\eta^2}{2}\mathcal{L}_{GD}^2
\right)[\tilde\rho_{t_k}]
 + O(\eta^3).
\end{align}
In comparison,
\begin{align}
\rho_{k+1} &= \frac{1}{m}\sum_{j=1}^m U_{k,j}\rho_k U_{k,j}^\dagger.\\
\end{align}
$\rho_0=\tilde\rho_0$ is smooth up to order $6$, and by Lemma \ref{lem:constant-Hamiltonian-smoothness-preservation} we know $\rho_k,k\geq 0$ is smooth up to order $6$. We describe the unitary with the notation $H_K(t)= e^{\psi(t)}\left(-\frac{1}{2}\Delta\right), H_{P,j}(t)=e^{\chi(t)}f_j$,
\begin{align}    
U_{k,j} &=
\exp\left(-i\frac{\eta}{2} H_K\left(t_k+\frac{\eta}{2}\right)\right)
\exp\left(-i \eta H_{P,j}\left(t_k+\frac{\eta}{2}\right)\right)
\exp\left(-i\frac{\eta}{2} H_K\left(t_k+\frac{\eta}{2}\right)\right).
\end{align}
Using Taylor series expansion,
\begin{align}
\frac{\d\rho_{k+1}}{\d \eta}\bigg|_{\eta=0}
&= [-iH(t_k),\rho_k], \\
\frac{\d^2\rho_{k+1}}{\d \eta^2}\bigg|_{\eta=0}
&=
[-i\dot{H}(t_k),\rho_k]
+\frac{1}{m}\sum_{j=1}^m\left[-iH_j(t_k),[-iH_j(t_k),\rho_k]\right], 
\end{align}
and $\frac{\d^3\rho_{k+1}}{\d \eta^3}\bigg|_{\eta=\xi}$ is the linear combination of nested commutators $[A_1,\cdots[A_n,\rho_k]\cdots]$ where $A_1,\ldots, A_n\in\{\nabla^2,f_1,\ldots,f_m\}$. By Lemma \ref{lem:constant-Hamiltonian-smoothness-preservation} we know $\rho_k$ is smooth up to order $6$, then by Lemma \ref{lem:nested-commutator-smoothness-preservation}, $\frac{\d^3\rho_{k+1}}{\d \eta^3}\bigg|_{\eta=\xi}=O(1)$.
Therefore,
\begin{align}\label{eqn:ts-expansion-dSQHD}
\rho_{k+1}
&=
\rho_k
+\eta \left[-iH(t_k),\rho_k\right]+\frac{\eta^2}{2}
\left(\left[-i\dot {H}(t_k),\rho\right]-\frac{1}{m}\sum_{j=1}^m\left[H_j(t_k),[H_j(t_k),\rho_k]\right]\right)
+ O(\eta^3)  \\
&= \Bigg[
\mI
+ \eta\left(\mL_{GD}+\frac{\eta}{2}\dot\mL_{GD}\right)
+ \eta^2 \left(\mL_{NOISE}+\frac{1}{2}\mL_{GD}^2\right)
\Bigg][\rho_{k}]
+ O(\eta^3) \\
&= \Bigg[
\mI
+ \eta\mL_{GD}
+ \frac{\eta^2}{2}\dot\mL_{GD}
+ \eta^2 \mL_{NOISE}+\frac{\eta^2}{2}\mL_{GD}^2
\Bigg][\rho_{k}]
+ O(\eta^3)
\end{align}
Let $e_k=\tilde\rho_{t_k}-\rho_k$, then
\begin{align}
e_{k+1} &= \Bigg[
\mI
+ \eta\mL_{GD}
+ \frac{\eta^2}{2}\dot\mL_{GD}
+ \eta^2 \mL_{NOISE}+\frac{\eta^2}{2}\mL_{GD}^2
\Bigg][e_{k}]
+ O(\eta^3),\\
&=\mR[e_k]+O(\eta^3).
\end{align}
Notice that $e_k$ is an operator smooth up to order $6$, and $\mR[e_k]$ is the linear combination of $\mB_1\mB_2[e_k]$ where $\mB_1,\mB_2\in\{\mI,\tilde\mL_{NOISE},\tilde\mL_K,\tilde\mL_{P}\}$. Therefore, by Lemma \ref{lem:channel-prod-bounded-derivative} for small enough $\eta$, there exists constant $C_1,C_2>0$ such that 
\begin{align}\label{eqn:recursive}
\tracenorm{\mR[e_k]} &\leq (1+C_1\eta)\tracenorm{e_k}, \\
\tracenorm{e_{k+1}}&\leq \tracenorm{\mR[e_k]}+C_2\eta^3.
\end{align}
Solve this recursive formula and take $T=N\eta$ as a constant, we have
\begin{align}
\tracenorm{e_k}
&\leq  \eta^3 C_2 \frac{(1+C_1\eta)^k-1}{(1+C_1\eta)-1}, \\
&\leq C_3 k\eta^3=C_3T\eta^2, \\
\tracenorm{e_k}&=O(\eta^2),
\end{align}
for some constant $C_3>0$.
% \sirui{formal definition of smooth quantum state such that $\tracenorm{\mL_{GD}[e_k]},\tracenorm{\mL_{NOISE}[e_k]},\tracenorm{\mL_{GD}^2[e_k]}$ are bounded; then complete this proof use Taylor series expansion and Gronwall’s Inequality} 
Thus we have proved for $k=0,1,\ldots,N-1$
\begin{align}
\tracenorm{\Lambda_{LS}(0,k\eta)[\rho_0] - \Lambda_{dSQHD}(k,\eta)[\rho_0]}=O(\eta^2).
\end{align}
\end{proof}

We extend the definition of the SQHD algorithm to allow for an adaptive learning rate.
\begin{definition}[Adaptive Stochastic Quantum Hamiltonian Descent]
Adaptive Stochastic Quantum Hamiltonian Descent with iteration number $N$ and adaptive learning rate $\vec\eta=[\eta_0,\ldots,\eta_{N-1}]$ is defined as
\begin{align}
U_{daSQHD}(N,\vec\eta;\xi)&=\prod_{j=N-1}^{0}\left[ \exp(-i \eta_j a_j (-\nabla^2/2)) \exp(-i\eta_j b_j f_{\xi_j})\right],
\end{align}
where $\xi$ is a $N$-dimensional random vector that $\xi_j$ is independently and uniformly drawn from $\{1,\ldots,m\}$ for $j=0,\ldots,N-1$. The corresponding channel is denoted as
\begin{align}
\Lambda_{daSQHD}(N,\vec\eta)[\rho] = \E_{\xi}\left[U_{daSQHD}(N,\vec\eta;\xi)\rho U^\dagger_{daSQHD}(N,\vec\eta;\xi)\right].
\end{align}
\end{definition}

The adaptive version of Theorem \ref{thm:order-2-approx} also holds.
\begin{theorem}\label{thm:order-2-approx-adaptive}
For any initial state $\rho_0$ such that $\tilde \rho_t=\Lambda_{LS}(u,0,t)[\rho_0]$ is smooth up to order $6$, any $e^{\psi(t)},e^{\chi(t)},u(t)$ smooth up to order $3$, $f_j(x),j=1,\ldots,m$ smooth up to order $6$ and small enough learning rate $\eta>0$, the process $\tilde \rho_t=\Lambda_{LS}(u,0,t)[\rho_0]$  is an order-$2$ quantum weak approximation of the Adaptive Stochastic Quantum Hamiltonian Descent process $\rho_k=\Lambda_{daSQHD} (k,\vec\eta)[\rho_0]$ where $\vec\eta=[u(\frac{\eta}{2})\eta,u(\frac{3\eta}{2})\eta,\ldots,u(\frac{(2N-1)\eta}{2})\eta]$.
\end{theorem}
\paragraph{Proof sketch} Let $\eta_k=u(t_k)\eta$. Eq.\eqref{eqn:ts-expansion-LS} becomes
\begin{align}
\tilde\rho_{t_{k+1}} &=
\left(\mI
+ \eta_k\left(\mathcal{L}_{GD}+\eta_k\mL_{NOISE}\right)
+ \frac{\eta^2}{2}\left(u(t_k)\dot{\mathcal{L}}_{GD}+\dot u(t_k){\mathcal{L}}_{GD}\right)
+ \frac{\eta_k^2}{2}\mathcal{L}_{GD}^2
\right)[\tilde\rho_{t_k}]
 + O(\eta^3),
\end{align}
and Eq.\eqref{eqn:ts-expansion-dSQHD} becomes
\begin{align}
\rho_{k+1}
&= \Bigg[
\mI
+ u(\frac{t_k+t_{k+1}}{2})\eta\mL_{GD}
+ \frac{u(\frac{t_k+t_{k+1}}{2})\eta^2}{2}\dot\mL_{GD}
+ (u(\frac{t_k+t_{k+1}}{2})\eta)^2 \mL_{NOISE}+\frac{(u(\frac{t_k+t_{k+1}}{2})\eta)^2}{2}\mL_{GD}^2
\Bigg][\rho_{k}]
+ O(\eta^3) \\
&= \Bigg[
\mI
+ \left(\eta_k +\dot u(t_k)\frac{\eta^2}{2}\right)\mL_{GD}
+ u(t_k)\frac{\eta^2}{2}\dot\mL_{GD}
+ \eta_k^2 \mL_{NOISE}+\frac{\eta_k^2}{2}\mL_{GD}^2
\Bigg][\rho_{k}]
+ O(\eta^3) \\
\end{align}
The two expansions match, and a recursive formula similar to Eq.\eqref{eqn:recursive} would prove the theorem.

\section{Numerical Experiment}
\label{sec:full-numerical}
\subsection{Implementation}
The numerical simulation is conducted on a space-discretized Hilbert space with $d\log_2 n_r$ qubits, where $n_r$ is the resolution of the grid. The computational basis in the Hilbert space is labeled by
\begin{align}
    \tilde\mC=\left\{[x_1,\ldots,x_d],x_j\in\left\{-1+\frac{2k+1}{n_r},k=0,1,\ldots,n_r-1\right\}\right\},
\end{align}
and the spacing is defined as $s=\frac{2}{n^r}$.

In this case, the simulated unitary is
\begin{align}
  & \tilde U_{dSQHD}(N,\eta,n_r;\xi) \\
= & \prod_{j=N-1}^{0}\left[ \exp(-i \frac{\eta}{2} a_j (-D_{d,n_r}/2)) \exp(-i\eta b_j F_{\xi_j,n_r})\exp(-i \frac{\eta}{2} a_j (-D_{d,n_r}/2))\right],
\end{align}
where
\begin{align}
F_{j,n} &= \sum_{x\in\tilde\mC} f_j(x)\ket{x}\bra{x}, \\
D_{d,n} &= \frac{1}{s^2}\sum_{j=0}^{d-1} \left(I_n^{\otimes j}\right)\otimes D_{1,n}\otimes \left(I_n^{\otimes d-1-j}\right), \\
(D_{1,n})_{ij} &= 2[j=i]-[j=(i+1)\bmod n]-[j=(i-1)\bmod n].
\end{align}
Notice that
\begin{align}
D_{1,n} &= U_{DFT,n}\Sigma_n U_{DFT,n}^\dagger, \\
\Sigma_n  &= \frac{1}{s^2}\textrm{diag}(4\sin^2(k\pi/n),k=0,\ldots,n-1), \\
(U_{DFT,n})_{jk} &= e^{-\frac{2\pi i}{n}jk},j,k=0,\ldots,n-1.
\end{align}
This allows a simplified unitary for $\tilde U_{dSQHD}$:
\begin{align}
  & \tilde U_{dSQHD}(N,\eta,n_r;\xi) \\
= & \prod_{j=N-1}^{0}\left[ \exp(-i \frac{\eta}{2} a_j (-D_{d,n_r}/2)) \exp(-i\eta b_j F_{\xi_j,n_r})\exp(-i \frac{\eta}{2} a_j (-D_{d,n_r}/2))\right] \\
= & U_{DFT,n_r}\prod_{j=N-1}^{0}\left[
\exp(-i \frac{\eta}{2} a_j (-\Sigma_{d,n_r}/2)) U_{DFT,n_r}^\dagger
\exp(-i\eta b_j F_{\xi_j,n_r})
U_{DFT,n_r}
\exp(-i \frac{\eta}{2} a_j (-\Sigma_{d,n_r}/2))
\right]
U_{DFT,n_r}^\dagger.
\end{align}
where
\begin{align}
U_{d,n_r} &= U_{DFT,n_r}^{\otimes d}, \\
\Sigma_{d,n_r} &= \frac{1}{s^2}\sum_{j=0}^{d-1} \left(I_n^{\otimes j}\right)\otimes \Sigma_{n}\otimes \left(I_n^{\otimes d-1-j}\right).
\end{align}
This makes the numerical simulation efficient because only (the exponential of) diagonal matrices and the discrete Fourier transform matrix are involved.
% Let the simulation cost of a single round in discrete SQHD be
% \begin{align}
%     2\cdot \#FFT +3\cdot \#DIAGONAL=O(n_r^d d\log n_r).
% \end{align}
% For $n_r=256=2^8$ and $d=2$, $\sim 2^{20}$ operations are involved in simulating one iteration.

\subsection{Setting}
\label{sec:app-ne-setting}

\paragraph{Test functions}
The first function we consider is the rotated high-dimensional double-well function:
\begin{align}
    F_{U,s}(\vec x)=F\left(\frac{1}{s}U\vec x\right), F(\vec x)=\frac{1}{d}\sum_{j=1}^d w(x_j),
\end{align}
where $U$ is a $d$-dimensional orthogonal matrix and $w:\mathbb{R}\to\mathbb{R}$ is a smooth non-convex function with $2$ local minima and $1$ global minimum. The objective function has $2^d$ local minima and $1$ global minimum. In \cite{QHD-separation-1}, the authors prove that for a class of well-formed $w(x)$, QHD finds an approximate solution $x$ in time polynomial to dimension $d$ and approximation error $\delta$.

In the experiment, we consider the case
\begin{align}
d&=2,s=1.2,\theta\sim\textrm{Uniform}[0,2\pi], \\
U&=\begin{bmatrix}
\cos\theta & \sin\theta \\
-\sin\theta & \cos\theta
\end{bmatrix},
\end{align}
with functions $w(x)=\frac{1}{10}(x^4-16 x^2+5 x)$. The function is the 2-dimensional Styblinski-Tang function (up to the rotation $U$), and it corresponds to the setting in \cite{QHD-separation-1}, and we denote it as \texttt{dw}.

The second function we consider is the Nonlinear Least Squares function: 
\begin{align}
    f(\vec x) = \frac{1}{n_{sample}}\sum_{j=1}^{n_{sample}}(h(\vec{x};\vec{a}_j)-b_j)^2,
\end{align}
where $h:\mathbb{R}^d\times\mathbb{R}^{\bar d}\to\mathbb{R}$ is a nonlinear function.
This objective arises naturally in many machine learning applications where the goal is to fit a model to data by minimizing the discrepancy between predicted and observed values \cite{deep-learning-goodfellow}. The formulation captures a broad class of supervised learning problems, where $\vec{x}$ represents learnable parameters, $\vec{a}_j$ are input features, $b_j$ are target outputs, and $h$ denotes a (potentially highly nonlinear) prediction function.

In the experiment, we consider the case
\begin{align}
d&=2,\bar d=3, \\
h(\vec{x},\vec{y})&=\sin^2\left((y_0+\sum_{k=1}^d y_k x_k)\right).
\end{align}
We select two sets of parameters.
In the first set, $n_{sample}=40$ and $\{\tilde a_1,\ldots,\tilde a_{20}\}$ is uniformly sampled from $\{\frac{0}{6\pi},\ldots,\frac{100}{6\pi}\}$ and $\{\tilde b_1,\ldots,\tilde b_{20}\}$ is uniformly sampled from $\{\frac{0}{4\pi},\ldots,\frac{100}{4\pi}\}$. We let
\begin{align}
\{a_{jk}\}=\{\{\tilde a_1,0,\tilde b_1\},\ldots,\{\tilde a_{20},0,\tilde b_{20}\},\{0,\tilde a_1,\tilde b_1\},\ldots,\{0,\tilde a_{20},\tilde b_{20}\}\},
\end{align}
and we denote it as \texttt{sino}.
In the second set, $n_{sample}=50$ and $a_{jk}$ is uniformly sampled from $\{-20/\pi,-19/\pi,\ldots,20/\pi\}$. We denote it as \texttt{sino-alt}.

We also consider the 2-dimensional Michalewicz function
\begin{align}
f(x_1,x_2) &= \frac{1}{2}\left(w(2x_1+2)+w(2x_2+2)\right),\\
w(x)&=-\sin(x)\sin(x^2/\pi)^{20},
\end{align}
which is denoted as \texttt{mich}, and the Cube-Wave function
\begin{align}
f(x_1,x_2) &= \frac{1}{2}\left(w(2x_1)+w(2x_2)\right),\\
w(x)&=\cos(\pi x)^2+\frac{1}{4}x^4,
\end{align}
which is denoted as \texttt{cubewave}. These two test functions come from \cite{gradient-QHD} (with rescaling).
All test functions we consider are smooth up to arbitrary order by  Definition~\ref{def:smooth}, which implies the generality of our result.
% \begin{table}[htbp]
%     \centering
%     \begin{tabular}{ccc}
%         \hline
%         Name & Smoothness bound $L_{\max}$ & Gradient noise $\sigma_f^*$  \\
%         \hline
%         \texttt{dw}  & 0 & 0 \\
%         \texttt{sino}  &  0 & 0 \\
%         \texttt{sino-alt}  &  0 & 0 \\
%         \texttt{mich}  &  0 & 0 \\
%         \texttt{cubewave}  &  0 & 0 \\
%         \hline
%     \end{tabular}
%     \caption{Parameters of the objective functions. These functions are all smooth up to arbitrary order, and we consider smoothness bound for order $2$ in this table.}
%     \label{tab:obj-func}
% \end{table}
The performance of algorithms on these test functions is assessed by the expected loss
\begin{align}
\E[f(x_N)-\inf f(x)],
\end{align}
and $\delta$-success probability
\begin{align}
\P\left[\frac{f(x_N)-\inf f(x)}{\sup f(x)-\inf f(x)}<\delta\right].
\end{align}
We set $\delta=0.01$ in \texttt{cubewave} and \texttt{dw}, $\delta=0.05$ in \texttt{sino-alt}, and $\delta=0.1$ in \texttt{mich} and \texttt{sino}.

\paragraph{Hamiltonian coefficients}
We consider two sets of Hamiltonian coefficients. The first one corresponds to NAGD:
\begin{align}\label{eqn:nagd-coeff}
e^{\psi(t)}=2t^{-3},e^{\chi(t)}=2t^3,u(t)=1.
\end{align}
And the second one corresponds to SGDM:
\begin{align}\label{eqn:sgdm-coeff}
e^{\psi(t)}=t^{-2},e^{\chi(t)}=2t,u(t)=\frac{1}{2}.
\end{align}
The first set promises an $O(t^{-2})$ convergence on QHD, and the second set promises an $O(t^{-1})$ convergence on SQHD.

\paragraph{Other settings}
The initial state for QHD and SQHD is set to $\frac{1}{|\tilde \mC|}\sum_{x\in\tilde\mC}\ket{x}$. For SGDM, the initial state $x_0$ is sampled from a uniform distribution on $\mC$.

The default total time and iteration number are $T=80, N=8000$ for all algorithms.
The default parameter for SQHD is SGDM Hamiltonian coefficients \eqref{eqn:sgdm-coeff}.
The default parameter for QHD is NAGD Hamiltonian coefficients \eqref{eqn:nagd-coeff}.
The parameter setting for SGDM is \cite{handbook}
\begin{align}
v_{k} & =\beta_{k} v_{k-1}+\nabla f_{j_k}(x_k), j_k\sim\textrm{Uniform}(1,\ldots,m),\\
x_{k+1} & =x_k-\gamma_k v_k,
\end{align}
where $x_0$ is given and $v_{-1}=0$, and $\beta_k=\frac{k}{k+2},\gamma_k=\frac{2\eta}{k+3}$ for $k=0,1,\ldots,N$.
For SGDM, the success probability and expectation of the objective function value are estimated using $1000$ independent runs. For SQHD, we use $10$ samples due to the high simulation cost of quantum algorithms. While more samples guarantee a more stable result, results on a few samples are enough to demonstrate the behavior of the SQHD algorithm in general.

\subsection{Validation of the approximation result}
We compare the result of the SQHD algorithm with direct simulation of Stochastic Quantum Hybrid Dynamics \eqref{eqn:LME-for-SQHD} in Figure \ref{fig:validation-thm-2} to validate the approximation result (Theorem \ref{thm:order-2-approx}).
\begin{figure}[htbp]
    \centering
    \includegraphics[width=\linewidth]{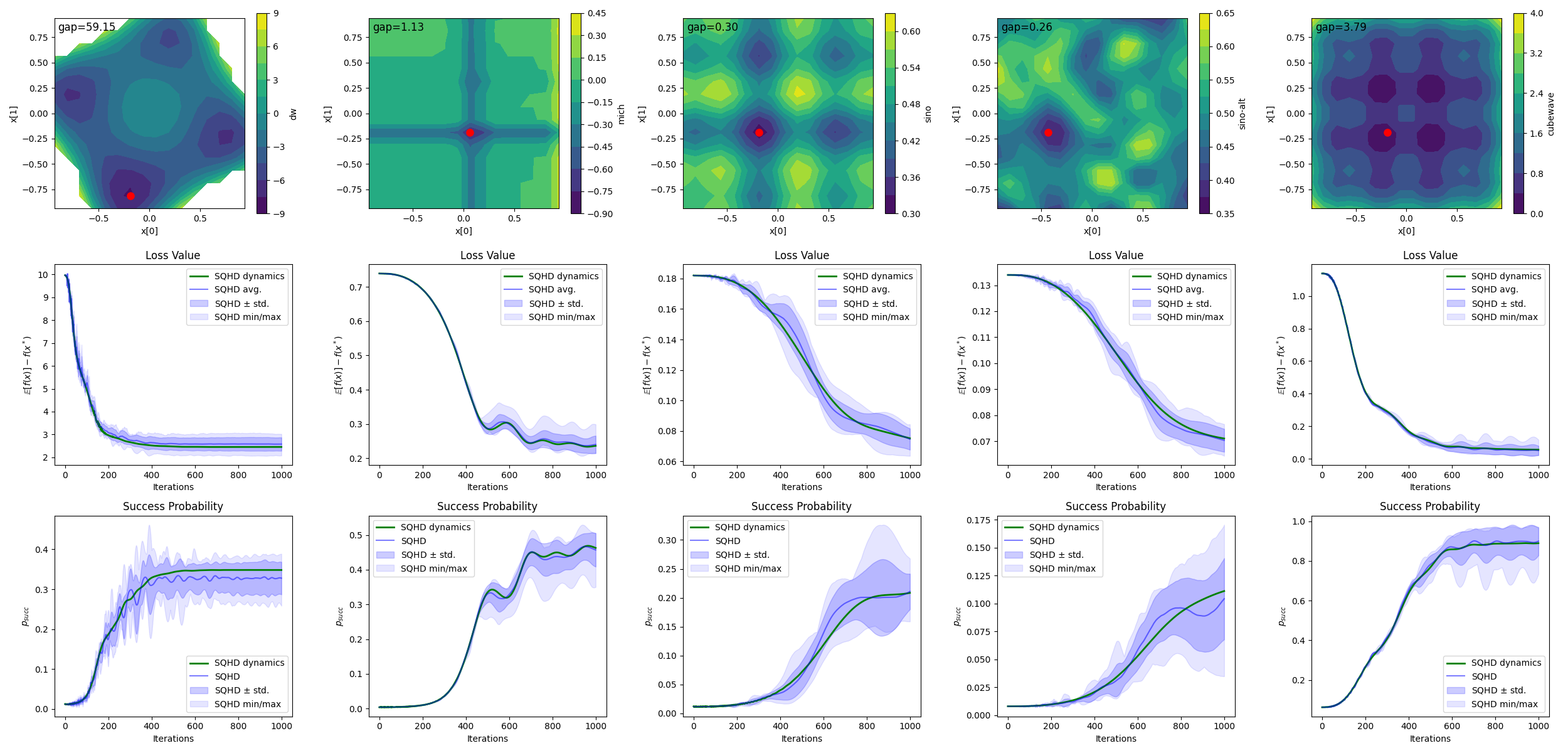}
    \caption{Results of Stochastic Quantum Hamiltonian Descent and Stochastic Quantum Hybrid Dynamics on all test functions with resolution $n_r=16$. The optimization parameter is set $T=10,N=1000$ with \texttt{SGDM-style} Hamiltonian coefficients.}
    \label{fig:validation-thm-2}
\end{figure}

\subsection{Rule of learning rate}
\label{sec:rule-of-lr}
We run SQHD, QHD, and SGDM with the default setting on all test functions with resolution $n_r=32$ except that the iteration number is set to $N=8000,16000,32000$. The results in Figure \ref{fig:rule-of-lr} show that SQHD with a smaller learning rate has smaller fluctuation.
The effect of learning rate is most obvious on the \texttt{sino-alt} function. 
For the default setting $N=8000,\eta=10^{-2}$, SQHD performs much worse than QHD. But as the learning rate becomes smaller, the performance of SQHD improves and surpasses the performance of QHD when $N=32000,\eta=2.5\times 10^{-3}$ regarding success probability. We speculate that the gradient noise $\sigma_f^*$ of the \texttt{sino-alt} objective function is relatively large, and
For the default setting, there is a mismatch between the gradient noise $\sigma_f^*$ and the learning rate $\eta$.
\begin{figure}[hp]
    \centering
    \includegraphics[width=0.65\linewidth]{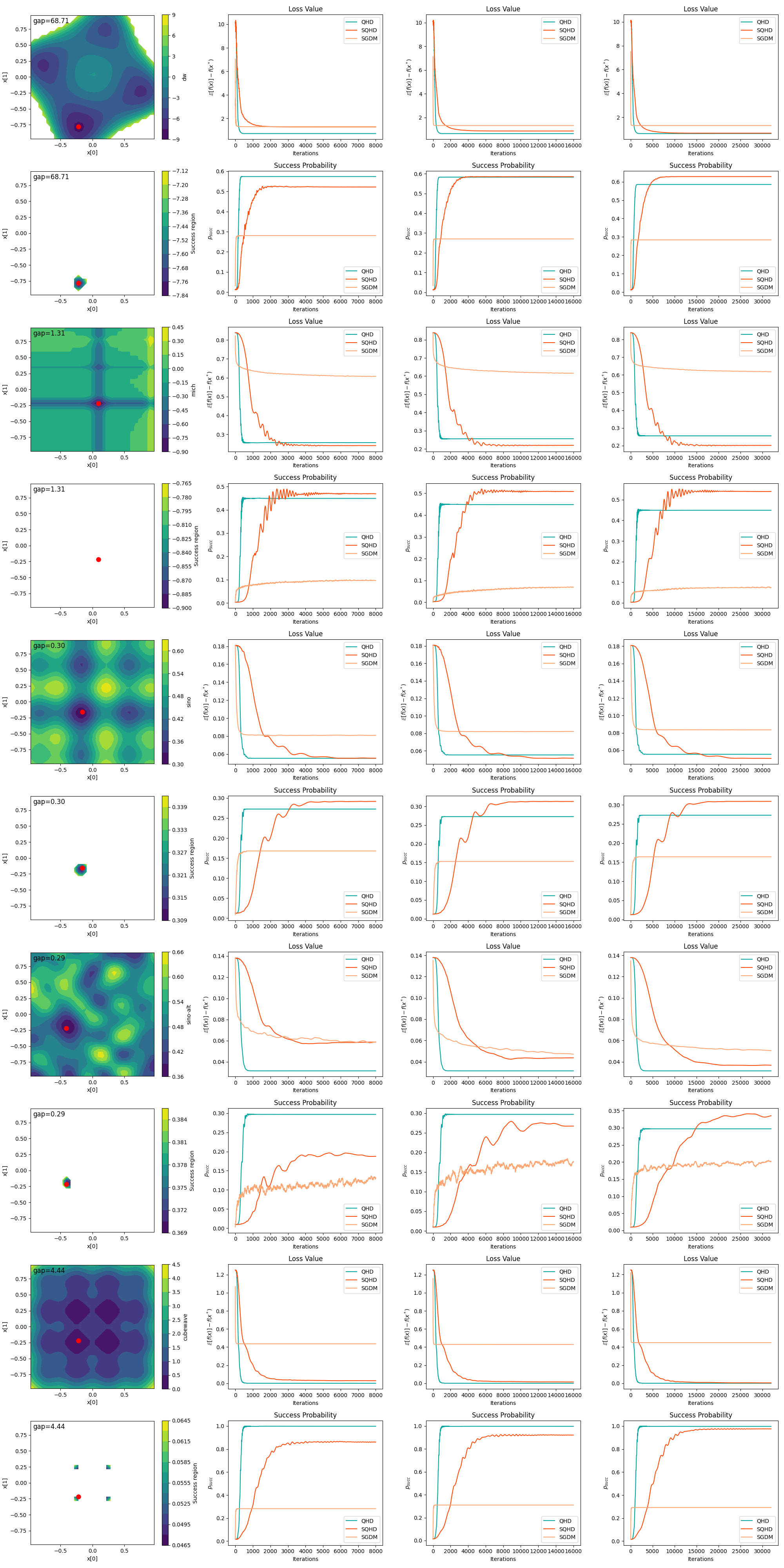}
    \caption{Results of SQHD, QHD, and SGDM on different objective functions varying with learning rate $\eta$. For each test function, the contour map and $\delta$-success region of the objective function are shown in the first column, and the second column shows the results with learning rate $\eta=0.01$, the third column for $\eta=0.005$, and the fourth column for $\eta=0.0025$.}
    \label{fig:rule-of-lr}
\end{figure}

\subsection{Rule of resolution}
We run QHD, SQHD, and SGDM with default settings except that the resolution is set to $n_r=32,128$. The results in Figure \ref{fig:rule-of-nres} show that the solution quality of SQHD is not affected by the change of resolution.
Nevertheless, SQHD shows a slower converging process at higher resolution.
\begin{figure}[hp]
    \centering
    \includegraphics[width=0.65\linewidth]{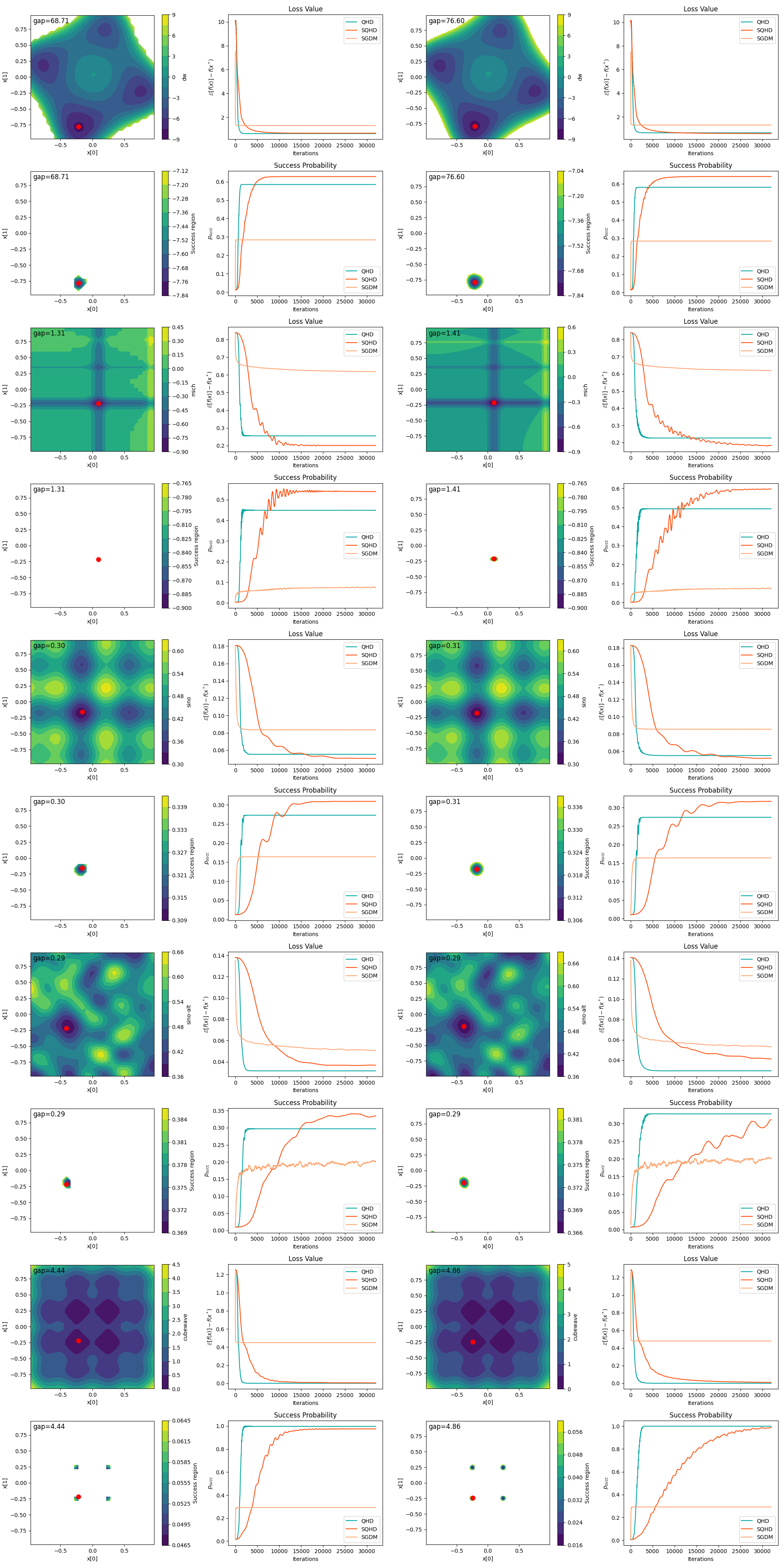}
    \caption{Results of SQHD, QHD, and SGDM on different objective functions varying with resolution $n_r$. For each test function,  the contour map and $\delta$-success region of the objective function are shown in the first column, and the second column shows the results. The last two columns show the same things, except that the resolution is $n_r=32$ for the first two columns and $n_r=128$ for the last two columns.}
    \label{fig:rule-of-nres}
\end{figure}

\subsection{General comparison}
\label{sec:general-comparison}
We compare the results of QHD, SQHD, and SGDM with default settings except for resolution $n_r=128$ and iteration number $N=32000$, and the results are shown in Figure \ref{fig:general-comparison}.
The results show that, overall, SQHD achieves comparable solution quality to QHD while incurring only a $1/m$ per-iteration computational cost, both outperforming SGDM in terms of solution quality.
One caveat is that under the default settings SQHD shows a slower convergence rate than QHD, which aligns with our theoretical expectations in the case of convex objective functions ($O(e^{-\beta(\int_0^t u(s)\d s)})$ versus $O(e^{-\beta(t)})$). For problems that require a long time to converge, SQHD’s advantage in per-iteration computational cost may be offset by its slower convergence rate. Identifying a separation between SQHD and QHD (as well as other classical continuous optimization algorithms) remains an open question.

\end{document}